\documentclass[12pt]{article}
\usepackage[a4paper]{geometry}
\usepackage[singlespacing]{setspace}
\usepackage[round]{natbib} 
\usepackage{amsmath,amsthm,epsfig,amsfonts, color, epstopdf,bm,mathrsfs}
\usepackage{amsbsy}
\usepackage{amssymb}

\usepackage{rotating}
\usepackage[font=footnotesize,labelfont=bf]{caption}
\usepackage[english]{babel}
\usepackage{graphicx,subfigure} % Allows including images
\usepackage{subcaption} % 用于创建子图
\usepackage{booktabs} % Allows the use of \toprule, \midrule and \bottomrule in tables
\usepackage{adjustbox}
\usepackage{caption}
\usepackage[font=footnotesize,labelfont=bf]{caption}
\usepackage{xcolor}
\usepackage{floatrow}
\usepackage{url}
\usepackage{afterpage}
\usepackage{geometry}
\usepackage{lipsum}
\usepackage{multirow}
\usepackage{authblk}
\usepackage[flushleft]{threeparttable}
\usepackage{hyperref}
\usepackage{pdflscape}
\usepackage{graphics}
\usepackage{graphicx}
\usepackage{algorithm}
\usepackage{algpseudocode}
\usepackage{enumerate}
\usepackage{tikz}
\usetikzlibrary{shapes.geometric, arrows}

\tikzstyle{startstop} = [rectangle, rounded corners, minimum width=3cm, minimum height=1cm,text centered, draw=black, fill=red!30, font=\footnotesize]
\tikzstyle{io} = [trapezium, trapezium left angle=70, trapezium right angle=110, minimum width=3cm, minimum height=1cm, text centered, draw=black, fill=blue!30, font=\footnotesize]
\tikzstyle{process} = [rectangle, minimum width=3cm, minimum height=1cm, text centered, draw=black, fill=orange!30, font=\footnotesize]
\tikzstyle{arrow} = [thick,->,>=stealth]

\newtheorem{definition}{{Definition}}
\newtheorem{proposition}{{Proposition}}
\newtheorem{theorem}{{Theorem}}
\newtheorem{lemma}{{Lemma}}

\newcommand{\uX}       {\mbox{\boldmath$X$}}

\algnewcommand{\LineComment}[1]{\State $\triangleright$ #1}
\algnewcommand{\NoNumberInput}{\item[\textbf{Input:}]}
\algnewcommand{\NoNumberOutput}{\item[\textbf{Output:}]}
\newcommand{\ubeta}             {\mbox{\boldmath$\beta$}}

\newcommand{\uiota}             {\mbox{\boldmath$\uiota$}}

\DeclareMathOperator*{\argmin}{argmin}

\def\l{\left}
\def\r{\right}
\def\sym#1{\ifmmode^{#1}\else\(^{#1}\)\fi}

\newfloatcommand{capbtabbox}{table}[][\FBwidth]

\def\marginnote#1{\setbox0=\vtop{\hsize4pc
		\small\raggedright\noindent\baselineskip9pt \rightskip=0.5pc plus
		1.5pc #1}\leavevmode \vadjust{\dimen0=\dp0
		\kern-\ht0\hbox{\kern-4.00pc\box0}\kern-\dimen0}}
\def\lboxit#1{\vbox{\hrule\hbox{\vrule\kern6pt
			\vbox{\kern6pt#1\kern6pt}\kern6pt\vrule}\hrule}}

\begin{document}
	
	\title{Robust transfer regression with corrupted labels}

	\author{{
			Sheng Pan\footnote{pansheng@staff.ynu.edu.cn.},
		}\\
		{\em\footnotesize  School of Mathematics and Statistics,Yunnan University, Kunming 650500, China}\\
		%{\em\footnotesize $^e$Yunnan Key Laboratory of Statistical Modeling and Data Analysis, Yunnan University, Kunming 650500, China}\\
		
	}
	
	\date{}
	
	\maketitle

	\begin{abstract}
		\noindent 
		In this paper, we introduce a robust transfer regression method designed to handle  corrupted labels in target data, under the scenarios that the corruption affects a substantial portion of the labels and the locations of these corruptions are unknown. Theoretical analysis substantiates our approach, illustrating that the estimation error consists of three components: the first relates to the source data; the second encompasses the domain shift ; and the third captures the estimation error attributed to the corrupted vector. Our theoretical framework ensures that the proposed method surpasses estimations based solely on target data. We validate our method through numerical experiments aimed at reconstructing corrupted compressed signals. Additionally, we apply our method to analyze the association between O6-methylguanine-DNA methyltransferase (MGMT) methylation and gene expression in Glioblastoma (GBM) patients.
		
		\vspace{0.5cm}
		
		signal recovery\\

	\end{abstract}
	
	%\newpage
	%\setcounter{page}{1}
	
	\section{Introduction}
	\paragraph{}
	In the field of data collection and analysis, data corruption refers to information that has been changed or damaged, leading to inaccuracies and reducing reliability. Especially in networked data compressive sensing (CS), it's not uncommon for a small number of sensors to report incorrect measurements or, in some cases, provide data points that are completely irrelevant(\cite{haupt2008compressed}). Similarly, real-world studies often face challenges from measurement errors such as misclassification and irregular assessment frequencies, which can harm the accuracy and credibility of research findings. While some inconsistencies are easy to spot and fix or remove during data cleaning, others fall within normal variability ranges, making them hard to detect (\cite{ackerman2024measurement}). These complexities require careful attention to maintain the integrity of conclusions based on the data.

	In this work, we focus on problems caused by corrupted labels. The corruption can be adversarial, covering scenarios like Huber's $\epsilon$-contamination model, which might affect a significant portion of the observations whose locations are unknown. In such cases, the data no longer follows an independently and identically distributed (i.i.d.) pattern, and the noise may not be symmetrically distributed.
	Traditional methods like Lasso \cite{tibshirani1996regression} and L1-Norm Quantile Regression \cite{li20081} do not perform well under these conditions. To tackle these challenges, several advanced methods have been developed for high-dimensional data.
	Extended Lasso techniques \cite{nguyen2012robust, descloux2022robust} aim to recover the true signal while also identifying error locations. The Median-of-Means approach \cite{Lecue2017RobustML, lecue2019learning, geoffrey2020robust} enhances robustness by dividing the dataset into smaller groups, calculating the mean for each group, and then taking the median of these means. This reduces the impact of outliers and heavy-tailed distributions.
	The robust gradient estimation method \cite{liu2019high, holland2019efficient} proposes estimating more reliable gradients during each iteration.
	These methods typically operate under the assumption that only target data is accessible for analysis. However, when source data is also available, transfer learning provides a potent alternative. By leveraging structural similarities across different but related domains or tasks, transfer learning has found successful application in numerous real-world scenarios. 
	
	In this paper, we focus on robust high-dimensional transfer regression. Various adaptation methods have been developed for transductive transfer learning, which can be applied to scenarios involving corrupted labels. 
	The marginal adaptation method proposed by\cite{pan2010domain} assumes that the conditional distributions of the target and source data are identical. Under this assumption, the Maximum mean discrepancy (MMD) introduced by Gretton et al.~\cite{gretton2012kernel} can be employed to measure the difference in predictor distributions between the two domains.
	Both joint distribution adaptation (JDA) by Long et al.~\cite{long2013transfer} and balanced distribution adaptation (BDA) by \cite{wang2017balanced} rely on the assumption that the MMD of class-conditional distributions can be approximated by replacing the true target labels with pseudo labels.
	However, these assumptions do not hold in our scenario, necessitating the development of alternative methods to achieve robust high-dimensional transfer regression. \cite{bastani2021predicting}, \cite{li2022transfer}, \cite{tian2023transfer} and \cite{li2024estimation} developed supervised transfer regression technique to improve the conventional estimation with L1 penalty. \cite{cai2024semi} proosed a semi-supervised triply robust inductive transfer learning under the assumption of scarce label of target data and   covariate shift.
	
	In this paper, we present a robust transfer Lasso algorithm specifically designed for signal reconstruction from potentially corrupted labels. Our contributions and findings can be summarized as follows:
	
	\begin{itemize}
		\item We propose a source data selection algorithm aimed at identifying suitable source datasets in the presence of potentially corrupted target data labels. By comparing the reconstructed signals derived from integrating target data with each source dataset to those obtained solely from the source datasets, our method effectively identifies and excludes source datasets that exhibit significant domain shifts.
		
		\item We present a transfer regression strategy designed for Lasso estimators that adjusts for both domain shifts and label corruption. Theoretical analysis reveals that the estimation error consists of three components: the Lasso estimation error on the aggregated selected source data, the impact of domain shift, and the estimation error due to label corruption. Furthermore, we establish the sign consistency property of our proposed algorithm.
		
		\item To validate our approach, we conducted numerical experiments focusing on the reconstruction of corrupted compressed signals. Notably, our method demonstrates a breakdown point exceeding 50\%. Additionally, we applied our method to explore the relationship between O6-methylguanine-DNA methyltransferase (MGMT) methylation and gene expressions in brain tissues of Glioblastoma (GBM) patients. Gene Ontology (GO) enrichment analysis of our results highlighted several pathways closely associated with GBM, underscoring the potential clinical relevance of our findings. 
	\end{itemize}

	The paper is outlined as follows: In Section 2.2, we first introduce the oracle version of robust transfer learning and present its theoretical results. In Section 2.3, we describe the source data selection process and detail the robust transfer Lasso algorithm. Sections 3 and 4 cover the simulations and the analysis of MGMT methylation and gene expression associations, respectively. All proofs are provided in the Appendix.

	\section{Methodology and main results\label{sec2}}
	\subsection{Notations}
	For a matrix $ A=\{a_{ij}\}_{i\in[n],j\in[p]} $, define the following norms:
	\begin{equation*}
		\|A\|_{\infty} = \max_{i,j}|a_{ij}|, \quad \|A\|_{L_1} = \max_i\sum_{j=1}^p|a_{ij}|.
	\end{equation*}
	
	For sets of indices $ T $ and $E$, $ A_{T} $ denotes the submatrix obtained by extracting those columns indexed by $ T $ and $ A_{r(E)} $ denotes the submatrix obtained by extracting those rows indexed by $ E $. $A_{ET}$ denotes the submatrix obtained by extracting those rows indexed by $ E $ and those columns indexed by $ T $. $\lambda_{\min}(A) $ represents the smallest eigenvalue of matrix $ A $; $ \text{diag}(A) $ denotes the vector composed of the diagonal elements of matrix $ A $.
	
	For a vector $ a=\{a_i\}_{i\in[n]} $, its norms are defined as follows:
	\begin{equation*}
		\|a\|_1 = \sum_{i\in[n]}|a_i|, \quad \|a\|_2 = \sqrt{\sum_{i\in[n]}a_i^2}, \quad \|a\|_0 = |\{j:a_j\neq 0\}|.
	\end{equation*}
	
	$ a_E $ denotes a vector where $ a_{E_j}=a_j $ if $ j \in E $, and $ a_{E_j}=0 $ if $ j \notin E $. $ a_{(E)} $ indicates extracting the elements indexed by $ E $. $HT_{\lambda}(a)=a\mathbf{1}\{|a_j|\geq \lambda\}$ denotes $a$ threshold at $\lambda$. $\mathrm{S}(a)=\{j:a_j\neq0\}$ denotes the support set of vector $a$.
	
	For a random sequence $ x_n $, $ x_n\xrightarrow{P}0 $ means that $ x_n $ converges in probability to 0 as $ n \to \infty $.
	
	\subsection{Oracle robust transfer Lasso}
	Throughout this article, we interpret the corruptions as follows:
	\begin{equation}\label{model0}
		Y_i = \uX_i^\top \ubeta^{\ast} + e_i^{\ast} + \varepsilon_i, \quad i=1,\ldots,n_0,
	\end{equation}
	where $\varepsilon_i$ denotes natural noise, and $e_i^{\ast}$ represents the corruption terms. Let $k := \|e^{\ast}\|_0$ denote the number of corrupted labels. Here, we assume that the source datasets are "clean," represented by:
	\begin{equation}\label{model1}
		Y_i^{(\mathrm{S}_j)} = \uX_i^{(\mathrm{S}_j)\top} \ubeta^{(\mathrm{S}_j)} + \varepsilon_i^{(\mathrm{S}_j)}, \quad i=1,\ldots,n_j, \; j=1,\ldots,L.
	\end{equation}
	This assumption can be validated using an extended Lasso method developed by \cite{nguyen2012robust}:
	$$
	\left(\hat{\ubeta}^{Rlasso}, \hat{e}^{Rlasso}\right) = \argmin_{e, \beta } \left\{ \frac{1}{2n_0}\|\mathbb{Y}-\mathbb{X}\beta- \sqrt{n_0}e \|_2^2 + \lambda_{1} \|\beta\|_1 + \lambda_{2} \|e\|_1 \right\},
	$$
	where $\lambda_{1},\lambda_{2}$ are hyperparameters. If the source data is "clean", then $|\{j: |\hat{e}^{Rlasso}_j| \geq \tilde{C}\sqrt{\log(n)/n}\}|$ should be small, where $\tilde{C}$ denotes some constant.
	 Rewriting Equations \eqref{model0} and \eqref{model1} in matrix form yields:
	\begin{align*}
		\mathbb{Y} &= \mathbb{X}\ubeta^{\ast} + e^{\ast} + \epsilon, \\
		\mathbb{Y}^{(\mathrm{S}_j)} &= \mathbb{X}^{(\mathrm{S}_j)}\ubeta^{(\mathrm{S}_j)} + \epsilon^{(\mathrm{S}_j)}, \quad j=1,\ldots,L.
	\end{align*}
	
	In this section, we introduce an oracle transfer regression algorithm designed to leverage source datasets sharing structural similarities with the target dataset. This approach is particularly beneficial when prior knowledge indicates which source datasets can provide valuable insights into the structure of the target data.
	
	Denote
	$$
	\mathcal{A}_h = \left\{ 1 \leq j \leq L : \| \Delta^{(S_j)} \|_1 \leq h \right\},
	$$
	where $\Delta^{(S_j)}$ represents the domain shift: $\ubeta^{\ast} - \ubeta^{(S_j)}$. Here, $\mathcal{A}_h$ is a set that includes indices $j$ of the source datasets where the domain shift $\Delta^{(S_j)}$  has an $L_1$ norm less than or equal to $h$. This set helps identify which source datasets are sufficiently similar to the target dataset in terms of their parameters. In the oracle scenario, we possess prior knowledge enabling the selection of source datasets that belong to this set.
	
	Upon selecting the appropriate source datasets, we aggregate the Lasso estimator using an iterative distributed calculation approach \cite{wang2017efficient}, which eliminates the need for a debiasing procedure. Denote 
	\begin{align*}
		\mathcal{L}_j(\ubeta) &= \frac{1}{2n_j}\left\|\mathbb{Y}^{(\mathrm{S}_j)} - \mathbb{X}^{(\mathrm{S}_j)}\ubeta\right\|_2^2.
	\end{align*}The details of the iterative distributed calculation are provided in Algorithm \ref{alg:edsl}.
	
	\begin{algorithm}
		\setstretch{0.7}
		\caption{Efficient Distributed Sparse Learning (EDSL)}
		\label{alg:edsl}
		\begin{algorithmic}[1]
			\Require Source datasets $\{\mathbb{X}^{(\mathrm{S}_i)}, \mathbb{Y}^{(\mathrm{S}_i)}\}_{i \in \mathcal{A}}$ and selected index set $\mathcal{A}$.
			\Ensure Distributed Lasso estimator $\hat{\ubeta}^{D(\mathcal{A})}$.
			\State \textbf{Initialization:} Select an element $v \in \mathcal{A}$. Compute the Lasso estimator $\hat{\ubeta}^{(\mathrm{S}_v)}$ on $\{\mathbb{X}^{(\mathrm{S}_v)}, \mathbb{Y}^{(\mathrm{S}_v)}\}$.
			\For{$t = 0, 1, \ldots$}
			\For{$j = 2, 3, \ldots, m$}
			\If{Receive $\widehat{\boldsymbol{\beta}}_t$ from the master}
			\State Calculate the gradient $\nabla \mathcal{L}_j(\widehat{\boldsymbol{\beta}}_t)$ 
			\EndIf
			\EndFor
			\State Update $\widehat{\ubeta}_{t+1}$ as follows:
			$$
			\widehat{\ubeta}_{t+1} = \arg\min_{\beta} \left\{ \mathcal{L}_{v}(\beta) + \left\langle \frac{1}{|\hat{\mathcal{A}}_h|} \sum_{j \in \hat{\mathcal{A}}_h} \nabla \mathcal{L}_j(\widehat{\ubeta}^{(t)}) - \nabla \mathcal{L}_{v}(\widehat{\ubeta}^{(t)}), \beta \right\rangle + \lambda_{t+1} \|\beta\|_1 \right\},
			$$
			where 
			$$
			\lambda_t = c_{\lambda,1}\sqrt{\frac{\log p}{\sum_{i \in \mathcal{A}_h} n_i}} + \sqrt{\frac{\log p}{n}} \left(c_{\lambda,2} s \sqrt{\frac{\log p}{n}} \right)^{t},
			$$
			with constants $c_{\lambda,1}, c_{\lambda,2}$.
			\EndFor
	
		\end{algorithmic}
	\end{algorithm}

The core idea of robust transfer regression lies in utilizing source data to enhance prediction accuracy when dealing with corrupted target data. This approach addresses two critical aspects: domain shift and data corruption, both of which are modeled as parametric components. The methodology proceeds by sequentially estimating these parameters, followed by the reconstruction of the target data signal through the integration of the aggregated estimated signal from source data and the computed domain shift. Given selected source data index $\mathcal{A}$ and hyperparameters $( \lambda_{\Delta}, \lambda_{e})$, the reconstructed signal is 
\begin{equation}\label{rg}
	\hat{\ubeta}(\mathcal{A}, \lambda_{\Delta}, \lambda_{e}) = \hat{\ubeta}^{D(\mathcal{A})} + \hat{\Delta}^{\mathcal{A}}( \lambda_{\Delta}, \lambda_{e}),
\end{equation}
where
$$
\l(\hat{\Delta}^{\mathcal{A}}( \lambda_{\Delta}, \lambda_{e}), \hat{e}^{\mathcal{A}}\r) = \argmin_{e, \Delta } \left\{ \frac{1}{2n_0}\|\mathbb{Y}-\mathbb{X}(\hat{\ubeta}^{D(\mathcal{A})}+\Delta)- \sqrt{n_0}e \|_2^2 + \lambda_{\Delta} \|\Delta\|_1 + \lambda_{e} \|e\|_1 \right\}.
$$

The selection of hyperparameters is performed adaptively. If the estimated fraction of corruptions is small, hyperparameters are selected using a cross-validation method on the target data. Conversely, if the estimated fraction of corruptions is large, hyperparameters are chosen based on selected validation data. The adaptive hyper-parameter selection algorithm is detailed in Algorithm \ref{alg:AHT}, where $c_h$ denotes the threshold for the fraction of corruptions, and $\tilde{c}$ is used to control the domain shift, ensuring that the reconstructed signal does not overfit the validation data.

To address potential non-sparsity in the aggregated Lasso estimator, a thresholding mechanism is applied during the final estimation phase. This step ensures the exclusion of extraneous variables that lie outside the support of the true signal. Thresholding for noise reduction is a well-documented practice in signal processing; for example, Donoho (1994) introduced "universal thresholds" set at $\sqrt{2\log n}$ for wavelet shrinkage.
The complete algorithmic implementation of this methodology is formally presented in Algorithm \ref{alg:oracle-trans-lasso}. 
\begin{algorithm}[H]
	\setstretch{0.9}
	\caption{Oracle Robust Transfer Lasso}
	\label{alg:oracle-trans-lasso}
	\begin{algorithmic}[1]
		\Require Target data $(\mathbb{X}, \mathbb{Y})$ and source datasets $\{\mathbb{X}^{(\mathrm{S}_i)}, \mathbb{Y}^{(\mathrm{S}_i)}\}_{i \in \mathcal{A}_h}$,threshold $c_h$, $\tilde{c}$, $\gamma_1$, fold number $k_0$
		\Ensure $\hat{\ubeta}^{\mathrm{oracle}}$, $HT_{\gamma_1}(\hat{\ubeta}^{\mathrm{oracle}})$
		\State \textbf{Aggregate Estimation on Source Data:} Compute the distributed Lasso estimator 
		$$
		\hat{\ubeta}^{D(\mathcal{A})} \leftarrow \mathrm{EDSL}\left(\{\mathbb{X}^{(\mathrm{S}_i)}, \mathbb{Y}^{(\mathrm{S}_i)}\}_{i \in \mathcal{A}_h}, \mathcal{A}_h\right).
		$$
		\State \textbf{Select hyperparameters by algorithm \ref{alg:AHT}: select $v$ with smallest domain shift,
			$$(\lambda_{\Delta},\lambda_e) \leftarrow  \mathrm{AHT}(c_h, \tilde{c}, k_0, \mathcal{A}_h,v).
			$$
			} 
		\State \textbf{Transfer Regression:}
			$$\hat{\ubeta}^{\mathrm{oracle}} \leftarrow \hat{\ubeta}(\mathcal{A}_h, \lambda_{\Delta}, \lambda_{e}).
		$$
		\State \textbf{Hard Thresholding:} Apply hard thresholding to obtain
		$$
		HT_{\gamma_1}(\hat{\ubeta}^{\mathrm{oracle}}) \leftarrow \hat{\ubeta}^{\mathrm{oracle}}.
		$$
	\end{algorithmic}
\end{algorithm}
\begin{algorithm}[H]
	\setstretch{0.9}
	\caption{Adaptive Hyper-parameter Tuning(AHT)}
	\label{alg:AHT}
	\begin{algorithmic}[1]
		\Require 
		\Statex Target data $(\mathbb{X}, \mathbb{Y})$ and source data $\{\mathbb{X}^{(\mathrm{S}_i)}, \mathbb{Y}^{(\mathrm{S}_i)}\}_{i \in \mathcal{A}}$, threshold $c_h$, $\tilde{c}$, fold number $k_0$, selected source data index $\mathcal{A}$, validation data index $v$.
		\Ensure 
		\Statex Optimal hyperparameters $(\lambda_{\delta},\lambda_e)$.
		\If{$|\{j:\hat{r}^{Rlasso}>\tilde{c}\log(n)/n\}| > c_h$}
		\State Choose the parameters $(\lambda_{\delta},\lambda_e)$ that minimize $\mathcal{L}_v(\hat{\ubeta}(\mathcal{A}, \lambda_{\Delta}, \lambda_{e}))+1000*\mathbb{I}\l(\| \hat{\ubeta}^{D(\mathcal{A})} - \hat{\ubeta}(\mathcal{A}, \lambda_{\Delta}, \lambda_{e})\|_1 > \tilde{c}\r)$.
		\Else
		\State Select $(\lambda_{\delta},\lambda_e)$ using $k_0$-fold cross-validation on target data.
		\EndIf
	\end{algorithmic}
\end{algorithm}
When the covariates follow a standard Gaussian distribution and no prior information about $\gamma_1$ is available, a feasible choice for threshold $\gamma_1$ of Algorithm~\ref{alg:oracle-trans-lasso} is $t_n$, as provided by Lemma~\ref{Lemma:3}:
\begin{align}\label{tn}
	\begin{split}
		t_n = (1+o(1))\bigg(9\hat{\sigma}_{\epsilon}\sqrt{\frac{\log p}{n_0}} &+ 12\hat{\sigma}_{\epsilon}\lambda_{t} + 3\lambda_{t}+ 4\hat{\sigma}_{\epsilon}\lambda_{\Delta} + \lambda_{\Delta}\bigg),
	\end{split}
\end{align}
where $\hat{\sigma}_{\epsilon}$ denotes a consistent estimator of $\sigma_{\epsilon}$.

For simplicity in the theoretical analysis and technical proofs, we assume that 
\[
n_0 = n_j = n, \quad j = 1, \dots, L.
\]
Before delving into the theoretical guarantees of the proposed algorithm, we first introduce several definitions that will be employed throughout the analysis.
		Denote
	\begin{itemize}
		\item $\bar{\ubeta}^{\mathcal{A}_h}=\sum_{j\in\mathcal{A}_h}\ubeta^{(\mathrm{S}_j)}/|\mathcal{A}_h|$, $\Delta^{\mathcal{A}_h}=\ubeta^*-\bar{\ubeta}^{\mathcal{A}_h}$
		\item $\bar{T}_h=\mathrm{S}(\bar{\ubeta}^{\mathcal{A}_h})$, $T=\mathrm{S}(\Delta^{\mathcal{A}_h})$
		\item 
		$s_{\Delta}=\|\Delta^{\mathcal{A}_h}\|_0$, $\bar{s}=\|\bar{\ubeta}^{\mathcal{A}_h}\|_0$
		\item $C_{\min}=\lambda_{\min}\left(\mathbb{X}_{T}^{\top} \mathbb{X}_{T}/n\right)$
		\item $\bar{C}_{\min}=\lambda_{\min}\left(\mathbb{X}^{\mathrm{S}_v\top}_{T} \mathbb{X}^{\mathrm{S}_v}_{T}/n\right)$
	\end{itemize}
	\begin{definition}[Extended Restricted Eigenvalue condition]\label{ere}
		A matrix $A$ satisfies the extended Restricted Eigenvalue (RE) condition if for  any sequences $a_n$:
		\begin{equation}
			\frac{1}{\sqrt{n}} \left\| A z + \sqrt{n}v \right\|_2 \geq \kappa_l (\left\| z \right\|_2 + \left\| v \right\|_2) + Ca_n\sqrt{\frac{\log p}{n}},
		\end{equation}
		for all vectors $ (z, v)$ and  $ \lambda>0$ satisfy
		\begin{equation}
			\left\| z_{T_0^c} \right\|_1 + \lambda \left\| v_{E^c} \right\|_1 \leq 3 \left\| z_{T_0} \right\|_1 + 3 \lambda \left\| v_E \right\|_1 +a_n,
		\end{equation}
		for any $T_0\subset [p]$, where $C$ is a universal constant, $E=\mathrm{S}(e^*)$.
	\end{definition}
		\begin{definition}[Mutual incoherence condition]\label{MI}
		A $n \times p$  matrix $A$ satisfies  mutual incoherence condition if there exists some $\gamma\in(0,1)$ such that,
	$$
	\max_{j \in T^c} \| (A_T^{\top} A_T)^{-1} A_T^{\top} \mathbf{a}_j \|_1 \leq 1 - \gamma, \quad 	\max_{j \in \bar{T}_h^c} \| (A_{\bar{T}_h}^{\top} A_{\bar{T}_h})^{-1} A_{\bar{T}_h}^{\top} \mathbf{a}_j \|_1 \leq 1 - \gamma,
	$$
		where $\mathbf{a}_j$ is the $j$-th column of $A$.
	\end{definition}
	\begin{definition}[Normalized columns]
	We assume that a $n \times p$ matrix $A$ has normalized columns, satisfying
	\begin{equation}
		\max_{j \in [p]} \frac{\|\mathbf{a}_j\|_2}{\sqrt{n}} \leq K_{\text{clm}},
	\end{equation}
	for some constant $K_{\text{clm}}$, where $\mathbf{a}_j$ is the $j$-th column of $A$.
   \end{definition}

	The following conditions are required for the asymptotic guarantees:
	\begin{enumerate}[{(C}1{)}]
		\item \textbf{Assumptions on covariates}: The $p$-dimensional covariates of both the target and source data are zero-mean sub-Gaussian random vectors sharing a common absolutely continuous distribution. The population covariance matrix $\Sigma$ has its smallest eigenvalue bounded away from zero and its largest eigenvalue bounded from above. These covariates have normalized columns and  the mean of each column is zero. Additionally, the design matrix $\mathbb{X}$ satisfies both the extended restricted eigenvalue condition and the mutual incoherence condition.
		
		\item \textbf{Assumptions on noise}: The noises $\epsilon_i, \epsilon_i^{(\mathrm{S}_j)}, i=1,...,n, j=1,...,L$ are zero-mean Gaussian variables with variance $\sigma_{\epsilon}$.
		
		\item \textbf{Assumptions on sample size}: As $n\to \infty$,
	
		$$
			|\mathcal{A}_h|h\sqrt{\frac{\log p\vee n}{n}}+ k\frac{\log(p\vee n)}{n}\to 0.
		$$
		
		\item \textbf{Assumptions on signal}:
		\begin{align*}
			\min_{j : \ubeta_j^* \neq 0} |\ubeta_j^*| \geq \gamma_1>0, \quad \min_{j : \ubeta_j^{(\mathrm{S}_j)} \neq 0} |\ubeta_j^{(\mathrm{S}_j)}|\geq \gamma_1>0
		\end{align*}
		for some constant $\gamma_1>0$.
	\end{enumerate}
	All these regularity assumptions are sufficiently general to apply to many real-world scenarios.
	
	For Condition C1, the extended Restricted Eigenvalue (RE) conditions can be satisfied in the case of Gaussian design; see \cite{nguyen2012robust} and \cite{raskutti2010restricted}. In the context of compressed sensing, the design matrix can be selected by the user. For other domains, a two-sample test technique developed in \cite{gretton2012kernel} can be employed to verify the distribution difference of covariates. Given two covariate datasets $\mathbb{X}$ and $\mathbb{X}^{(\mathrm{S}_i)}$, the Maximum Mean Discrepancy (MMD) proposed in \cite{gretton2012kernel} is defined as:
	\begin{equation}
		\mathrm{MMD} \left[ \mathcal{F}, \mathbb{X}, \mathbb{X}^{(\mathrm{S}_i)} \right] := \sup_{f \in \mathcal{F}} \left( \frac{1}{n} \sum_{i=1}^{n} f(\mathbb{X}_i) - \frac{1}{n} \sum_{i=1}^{n} f(\mathbb{X}^{(\mathrm{S}_i)}_i) \right),
	\end{equation}
	where $\mathcal{F}$ represents the unit ball in a reproducing kernel Hilbert space. If the marginal distributions of the covariates are similar, then the MMD should be small. The mutual incoherence condition can be satisfied if the columns of the covariates are nearly orthogonal.
	
	Condition C2 holds approximately if the distribution of natural noise is symmetric, not heavy-tailed, and does not contain outliers. Condition C3 specifies the sample size requirement for signal recovery. While the true signal may not be sparse, Condition C4 involves a sparse approximation of the true signal. This assumption is not unrealistic; for example, in many image processing applications, the gray levels of pixels belonging to an object are significantly higher than those of background pixels \cite{sezgin2004survey}.

	In the following lemma, we provide the $l_{\infty}$ bound for $\hat{\ubeta}^{\mathrm{oracle}}-\ubeta^*$ and $\hat{e}^{\mathcal{A}_h}-e^*$.
	
	\begin{lemma}\label{Lemma:3}
		Assume conditions (C1)-(C4) hold. Further suppose that 
		$$
		\lambda_{\Delta} = \frac{2\|\mathbb{X}^{\top}\epsilon\|_{\infty}}{n}, \quad \lambda_{e} = \frac{2\|\epsilon\|_{\infty}}{\sqrt{n}},
		$$
		then with probability approaching 1 as $n \to \infty$,
	\begin{align}\label{lem3:r1}
			\setstretch{0.9}
		\begin{split}
			&\|\hat{\ubeta}^{\mathrm{oracle}}-\ubeta^*\|_{\infty} \leq \mathcal{T}_1 +\mathcal{T}_2+\mathcal{T}_3 ,
		\end{split}
	\end{align}
		where 
		\begin{align*}
		\mathcal{T}_1=&9\sigma_{\epsilon}K_{clm}\sqrt{\frac{\log p}{n}}+12\frac{1}{\sqrt{\bar{C}_{\min}}}\sigma_{\epsilon}\lambda_{t}+ 3\left\|\left(\frac{\mathbb{X}_{\bar{T}_h}^{\mathrm{S}_v\top} \mathbb{X}_{\bar{T}_h}^{\mathrm{S}_v}}{n}\right)^{-1}\mathrm{sign}(\bar{\ubeta}^{(t)}_{(\bar{T}_h)})\right\|_{\infty}\lambda_{t}\\
		\mathcal{T}_2=&4\frac{1}{\sqrt{C_{\min}}}\sigma_{\epsilon}\lambda_{\Delta} + \left\|\left(\frac{\mathbb{X}_{T}^\top \mathbb{X}_{T}}{n}\right)^{-1}\mathrm{sign}(\Delta^{\mathcal{A}_h}_{(T)})\right\|_{\infty}\lambda_{\Delta} \\
		\mathcal{T}_3=& \left\|\left(\frac{\mathbb{X}_{T}^{\top} \mathbb{X}_{T}}{n}\right)^{-1}\right\|_{L_1}\frac{\|\Sigma\|_{\infty}\sqrt{\log(s_{\Delta}n)}}{n} \times \\
		&\qquad \left[4\max\left(\sqrt{s_{\Delta}}\lambda_{\Delta}/\lambda_e,\sqrt{k}\right) \left(\sqrt{\frac{s_{\Delta}\log p}{n}}+\sqrt{\frac{k\log n}{n}}\right)\right. \\
		&\qquad \left.+ O\l(\frac{\bar{s}\log p}{|\mathcal{A}_h|n}\r)\right].
		\end{align*}
		 If additional suppose that
	  $
	  \|\Omega_{TT}\|_{L_1}=O(1),
	  $
	  where $\Omega_{TT}$ is e the inverse of $\Sigma_{TT}$, then
		
		\begin{align*}
			\|\hat{e}^{\mathcal{A}_h}-e^*\|_{\infty}=O_P\left( \sqrt{\frac{\log n}{n}}\right).
		\end{align*}
	\end{lemma}
	As a consequence of Lemma \ref{Lemma:3}, we may be able to establish the support recovery property of $\hat{\ubeta}^{\mathrm{oracle}}$, as the following proposition states.
	
\begin{proposition}\label{prop0}
	Under conditions of Lemma \ref{Lemma:3},
	it follows that
	\begin{equation*}
		\|\hat{\ubeta}^{\mathrm{oracle}} - \ubeta^*\|_2 + \|\hat{e}^{\mathcal{A}_h} - e^*\|_2 = O_P\left( \sqrt{\frac{\bar{s}\log p}{|\mathcal{A}_h|n}} + h \wedge \sqrt{\frac{s_{\Delta}\log p}{n}} + \sqrt{\frac{k\log n}{n}}\right),
	\end{equation*}
	where $s_{\Delta} = \|\bar{\Delta}^{\mathcal{A}_h}\|_0$ and $\bar{s} = \|\bar{\ubeta}^{\mathcal{A}_h}\|_0$. Furthermore,
	\begin{equation*}
		P\left(\mathrm{sign}\left(HT_{\gamma_1}\left(\hat{\ubeta}^{\mathrm{oracle}}\right)\right) = \mathrm{sign}(\ubeta^*)\right) \to 1.
	\end{equation*}
	If the covariates follow a standard Gaussian design, then
	\begin{equation*}
	P\left(\mathrm{sign}\left(HT_{t_n}\left(\hat{\ubeta}^{\mathrm{oracle}}\right)\right) = \mathrm{sign}(\ubeta^*)\right) \to 1.
    \end{equation*}
	where $t_n$ is defined as in (\ref{tn}).
\end{proposition}
	
	From the results of Lemma \ref{Lemma:3} and Proposition \ref{prop0}, we observe that the estimation error comprises three components: the first is associated with the source data; the second involves the domain shift and its estimation error; and the last term represents the estimation error of the corrupted vector. Theoretical results indicate that if the sparsity pattern in the source data does not closely resemble that of the target data, $s_{\Delta}$ and $h$ will be significantly larger, potentially leading to negative transfer. In Section 2.3, we will offer an algorithm to make source data selection to ensure small $s_{\Delta}$ and $h$.
	
	\subsection{Robust transfer Lasso}
	In many scenarios, knowledge regarding which source data shares a similar structure with the target data is unavailable. Consequently, source data selection becomes a critical step in transfer regression. Under Condition C4, $s_{\Delta}$ can be regulated by $h$. Therefore, it is essential to ensure that the domain shift $\|\Delta^{(j)}\|_1 := \|\ubeta^* - \ubeta^{(\mathrm{S}_j)}\|_1$ for each selected source dataset is minimal.
	To estimate $\|\Delta^{(j)}\|_1$, for $j=1,\ldots,L$, a straightforward approach involves computing $\|\hat{\ubeta}^{(\mathrm{S}_j)} - \hat{\ubeta}^{\mathrm{(rlasso)}}\|_1$, where $\hat{\ubeta}^{(\mathrm{S}_j)}$ are the Lasso estimators of $\ubeta^{(\mathrm{S}_j)}$. However, this method may suffer from significant bias.
	
	To mitigate this issue, we first estimate $\ubeta^{(j)} := (\ubeta^{(\mathrm{S}_i)} + \ubeta^*)/2$ by combining two datasets:
	\begin{equation}\label{v1}
		(\hat{\ubeta}^{(j)},\hat{e}^{(j)}) = \argmin_{ e , \beta} \left\{ \frac{1}{2(n_0+n_j)}\left\|\mathbb{Y}^{(j)}-\mathbb{X}^{(j)}\beta -\sqrt{n_0+n_j}e\right\|_2^2 + \lambda_{\beta}^{(j)} \|\beta\|_1+\lambda_{e}^{(j)} \|e\|_1 \right\}.
	\end{equation}
	where 
	$$
	\mathbb{Y}^{(j)} = 
	\begin{bmatrix}
		\mathbb{Y}^{(\mathrm{S}_i)} \\
		\mathbb{Y}
	\end{bmatrix}
	,\quad
	\mathbb{X}^{(j)} = 
	\begin{bmatrix}
		\mathbb{X}^{(\mathrm{S}_i)} \\
		\mathbb{X}
	\end{bmatrix}.
	$$
	We then estimate $\|\Delta^{(j)}\|_1$ by $\hat{h}_j := 2\left\|\hat{\ubeta}^{(j)} - \hat{\ubeta}^{(\mathrm{S}_j)}\right\|_1$. By merging two datasets, the estimation bias can be reduced due to the larger sample size, compared to $\|\hat{\ubeta}^{(\mathrm{S}_j)} - \hat{\ubeta}^{(\mathrm{rlasso})}\|_1$. This allows us to propose the Source Data Selection (SDS) algorithm \ref{SDS}.
	
	The performance of the robust transfer Lasso algorithm is sensitive to hyperparameters. To address this, we select a validation dataset during the source data selection procedure for tuning these parameters. The detailed robust transfer Lasso (RTL) algorithm is outlined in Algorithm \ref{alg:robust-trans-lasso}.
\begin{algorithm}[H]
	\setstretch{0.9}
	\caption{Source Data Selection (SDS)}
	\label{SDS}
	\begin{algorithmic}[1]
		\Require 
		\Statex Target data $(\mathbb{X}, \mathbb{Y})$ and source data $\{\mathbb{X}^{(\mathrm{S}_i)}, \mathbb{Y}^{(\mathrm{S}_i)}\}_{i \in \mathcal{A}}$, $h$, $A$
		\Ensure 
		\Statex $\hat{\mathcal{A}}_h$, validation dataset $\hat{v}$
		
		\State \textbf{Select Source Data with Small Domain Shift}: 
		$$
		\hat{\mathcal{A}}_h \leftarrow \left\{ j \in \mathcal{A} \mid \hat{h}_j \text{ is among the smallest } A \text{ values} \right\} \cap \{ j \in \mathcal{A} : \hat{h}_j \leq h \}.
		$$
		
		\State \textbf{Select Validation Dataset}: 
		$$
		\hat{v} \leftarrow \arg\min_{1 \leq j \leq L} \hat{h}_j.
		$$
	\end{algorithmic}
\end{algorithm}

\begin{algorithm}[H]
	\setstretch{0.8}
	\caption{Robust Transfer Lasso Algorithm}
	\label{alg:robust-trans-lasso}
	\begin{algorithmic}[1]
		\Require 
		\Statex  $(\mathbb{X}, \mathbb{Y})$, $\{\mathbb{X}^{(\mathrm{S}_i)},  \mathbb{Y}^{(\mathrm{S}_i)}\}_{i \in [L]}$, threshold $c_h$, $\tilde{c}$, $\gamma_1$, fold number $k_0$
		\Ensure 
		\Statex $\hat{\ubeta}^{\hat{\mathcal{A}}_h},\quad HT(\hat{\ubeta}_{\gamma_1}^{\hat{\mathcal{A}}_h})$
		
		\State \textbf{Source Data Selection}: Perform source data selection:
		$$
		(\hat{\mathcal{A}}_h,\hat{v}) \leftarrow \mathrm{SDS}\left(\{\mathbb{X}^{(\mathrm{S}_i)}, \mathbb{Y}^{(\mathrm{S}_i)}\}_{i \in [L]}\cup (\mathbb{X}, \mathbb{Y})\right).
		$$
	\State \textbf{Aggregate Estimation on Source Data:} Compute the distributed Lasso estimator:
	$$
	\hat{\ubeta}^{D(\hat{\mathcal{A}}_h)} \leftarrow \mathrm{EDSL}\left(\{\mathbb{X}^{(\mathrm{S}_i)}, \mathbb{Y}^{(\mathrm{S}_i)}\}_{i \in \hat{\mathcal{A}}_h},\hat{\mathcal{A}}_h\right).
	$$
	\State \textbf{Select hyperparameters by algorithm \ref{alg:AHT}:
		$$(\lambda_{\Delta},\lambda_e) \leftarrow  \mathrm{AHT}(c_h, \tilde{c}, k0, \hat{\mathcal{A}}_h,\hat{v}).
		$$
	} 
	\State \textbf{Transfer Regression: Using (\ref{rg}),}
	$$\hat{\ubeta}^{\hat{\mathcal{A}}_h} \leftarrow \hat{\ubeta}(\hat{\mathcal{A}}_h, \lambda_{\Delta}, \lambda_{e}).
	$$
	\State \textbf{Hard Thresholding:} Apply hard thresholding to obtain
	$$
	HT_{\gamma_1}(\hat{\ubeta}^{\hat{\mathcal{A}}_h}) \leftarrow \hat{\ubeta}^{\hat{\mathcal{A}}_h}.
	$$
	\end{algorithmic}
\end{algorithm}
		Without prior information about $\gamma_1$, a possible choice of $\gamma_1$ in the step 5 of algorithm \ref{alg:robust-trans-lasso} is  $t_n$  defined as in (\ref{tn}).
	
	\begin{lemma}\label{lem:tr}
		Under conditions C1-C4, further suppose that
		\begin{align*}
		\lambda_{\beta}^{(j)}&\geq\frac{1}{n} \|\mathbb{X}^\top \epsilon\|_\infty + \frac{1}{n} \|\mathbb{X}^{(\mathrm{S}_j)\top} \epsilon^{(\mathrm{S}_j)}\|_\infty ,\\
		\lambda^{(j)}_{e}&\geq \frac{2}{\sqrt{n}} \|\epsilon\|_\infty+ \frac{1}{2\sqrt{n}}\l(\|\mathbb{X}\|_{\infty}+\|\mathbb{X}^{(\mathrm{S}_j)}\|_{\infty}\r)\|\Delta^{(\mathrm{S}_j)}\|_1,
		\end{align*}
		then,
		$$
		|\hat{h}_j - \|\Delta^{(j)}\|_1| = O_P\l(s_{j,0}\vee k\sqrt{\frac{\log p}{\log n}}\sqrt{\frac{\log (np)}{n}}\|\Delta^{(\mathrm{S}_j)}\|_1 +s_{j,0}\sqrt{\frac{\log p}{n}}\r),
		$$
		for $j=1,...,L$, where $\Delta^{(\mathrm{S}_j)} = \ubeta^{(\mathrm{S}_j)} - \ubeta^*$ and $s_{j,0} = \| \Delta^{(\mathrm{S}_j)} \|_0$.
	\end{lemma}
	The theoretical results of Lemma \ref{lem:tr} indicate that the estimation error associated with domain shift is influenced by both the corruption fraction and the true domain shift. Given prior knowledge of the corruption fraction, we recommend selecting a relatively large value of $h$ in Algorithm \ref{SDS} to mitigate these effects.
	\begin{theorem}\label{th:1}
		Assume that the conditions of Proposition \ref{prop0}  hold. Then,
		$$
		P\left(\hat{\mathcal{A}}_h \subset \mathcal{A}_h\right) \to 1,
		$$
		and on the event $\{|\hat{\mathcal{A}}_h| > 1\}$, it holds with probability approaching 1 as $n \to \infty$,
		$$
		\|\hat{\ubeta}^{\hat{\mathcal{A}}_h} - \ubeta^*\|_2 + \|\hat{e}^{\hat{\mathcal{A}}_h} - e^*\|_2 \leq C \left( \sqrt{\frac{\bar{s}^{\hat{\mathcal{A}}_h} \log p}{|\hat{\mathcal{A}}_h| n}} + h \wedge \sqrt{\frac{s^{\hat{\mathcal{A}}_h}_{\Delta} \log p}{n}} + \sqrt{\frac{k \log n}{n}} \right),
		$$
		for some universal constant $C$, where $\bar{s}^{\hat{\mathcal{A}}_h} = \left\|\sum_{j \in \hat{\mathcal{A}}_h} \ubeta^{(\mathrm{S}_j)}\right\|_0$ and $s^{\hat{\mathcal{A}}_h}_{\Delta} = \left\|\sum_{j \in \hat{\mathcal{A}}_h} \frac{\ubeta^{(\mathrm{S}_j)}}{|\hat{\mathcal{A}}_h|} - \ubeta^*\right\|_0$. Furthermore,
		\begin{equation*}
			P\left(\mathrm{sign}\left(HT_{\gamma_1}\left(\hat{\ubeta}^{\hat{\mathcal{A}}_h} \right)\right) = \mathrm{sign}(\ubeta^*)\right) \to 1.
		\end{equation*}
		If the covariates follow a standard Gaussian design, then
		\begin{equation*}
	P\left(\mathrm{sign}\left(HT_{t_n}\left(\hat{\ubeta}^{\hat{\mathcal{A}}_h} \right)\right) = \mathrm{sign}(\ubeta^*)\right) \to 1.
     \end{equation*}
		where $t_n$ is defined as in (\ref{tn}).
	\end{theorem}

	\section{Simulation studies\label{sec3}}
	In this section, we present simulations that illustrate the robustness and effectiveness of our proposed method for recovering a sparse signal from corrupted compressive samples. For comparison, we evaluate our method against several established techniques, including Lasso \cite{tibshirani1996regression}, robust Lasso (Rlasso) \cite{nguyen2012robust}, and Transfer Lasso \cite{li2022transfer}.
	
	The target data is generated based on a synthetic 12-sparse signal with $p=400$ features and $n=100$ observations, as depicted in Figure~\ref{fig:a}. Corruption is introduced following a uniform distribution $U[0.5, 1]$, with the fraction of corruptions $r=k/n$ varying from 10\% to 90\% across simulations. All noise in both target and source datasets is generated according to a normal distribution $N(0, 0.01)$.
	
	Specifically, the source data are generated as follows:
	\begin{equation}
		y^{(\mathrm{S}_j)}_i = \mathbf{X}^{(\mathrm{S}_j)}_i \boldsymbol{\beta}^{(\mathrm{S}_j)} + \epsilon_{i}^{(\mathrm{S}_j)}, \quad i=1,\dots,n_j, \quad j=1,\dots,L,
	\end{equation}
	with $n_j=100$, $j=1,...,L$, and $\|\boldsymbol{\beta}^{(\mathrm{S}_j)}\|_0 = 12U + (1-U)20$, where $U$ is a binomial variable such that $P(U=1)=1-\frac{1}{L}$. The common structure, defined as the cardinality of the set $\{j: \boldsymbol{\beta}^{(\mathrm{S}_j)}_j = \boldsymbol{\beta}_j^* \neq 0\}$, varies from 4 to 12. Additionally, $\Delta_j$, for $j=1,\dots,L$, varies within the range [2, 24]. All sensing matrices are generated using a  Gaussian distribution with covariance matrix $\mathbb{I}/\sqrt{n}$.
	
	The optimization problems described in Algorithms~\ref{alg:edsl}--\ref{alg:robust-trans-lasso} and Rlasso are convex and can be  solved using the Alternating Direction Method of Multipliers (ADMM) \cite{boyd2011distributed}, coordinate descent, or interior-point methods \cite{boyd2004convex}. For implementing the Lasso method, we utilize the well-established R package \texttt{glmnet}. The hyperparameters for the Lasso method are selected via cross-validation. In the case of the Rlasso method, hyperparameter selection is guided by recommendations from \cite{nguyen2012robust}. Additionally, the implementation of the Transfer Lasso method is based on the code provided in \cite{li2022transfer}.
	
	The reconstruction of sparse signals using various methodologies is illustrated in Figure~\ref{f1}. It is observed that the Lasso method becomes ineffective upon the introduction of corrupted data, while our proposed approach maintains robustness even at high levels of corruption. In Figure \ref{fig:ser}, the performance evaluation is conducted through the mean performance Signal-to-Error Ratio (SER) [dB] obtained from 1000 simulations. The SER [dB] is defined as:
	\begin{equation}
		\text{SER}(\mathbf{x},\hat{\mathbf{x}})[\mathrm{dB}] = 10 \log_{10} \left( \frac{\sum_{i=1}^{p} x_i^2}{\sum_{i=1}^{p} (x_i - \hat{x}_i)^2} \right),
	\end{equation}
	where $\hat{\mathbf{x}}=\{\hat{x}_i\}_{i=1}^{p}$ denotes the reconstructed signal and $\mathbf{x}=\{x_i\}_{i=1}^{p}$ represents the true signal. Higher SER values indicate superior performance. Remarkably, our method demonstrates a breakdown point exceeding 50\%. Specifically, in contexts with a minimal fraction of corruption, our approach achieves SER values comparable to those of the oracle case.

	\begin{figure}
		\centering     %%% not \center
		\subfigure[Noiseless]{\label{fig:a}\includegraphics[width=50mm]{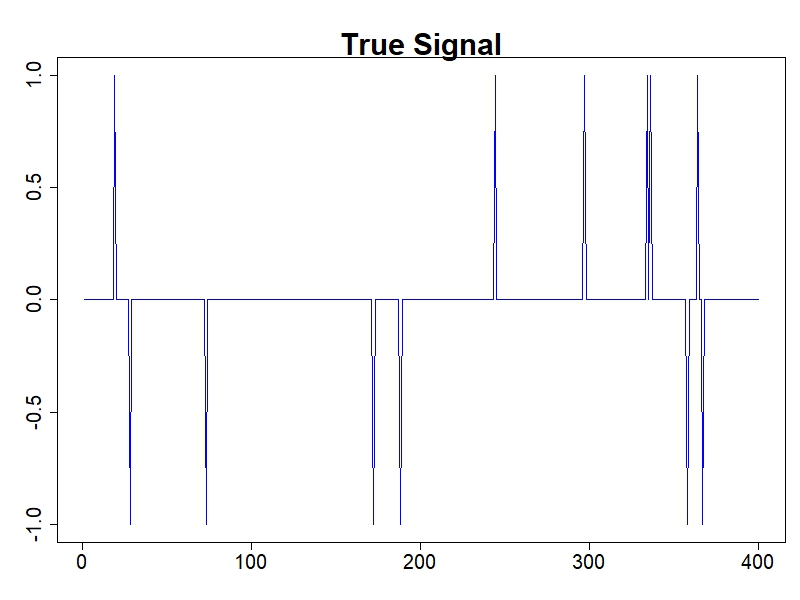}}
		\subfigure[$r$=0]{\label{fig:b1}\includegraphics[width=50mm]{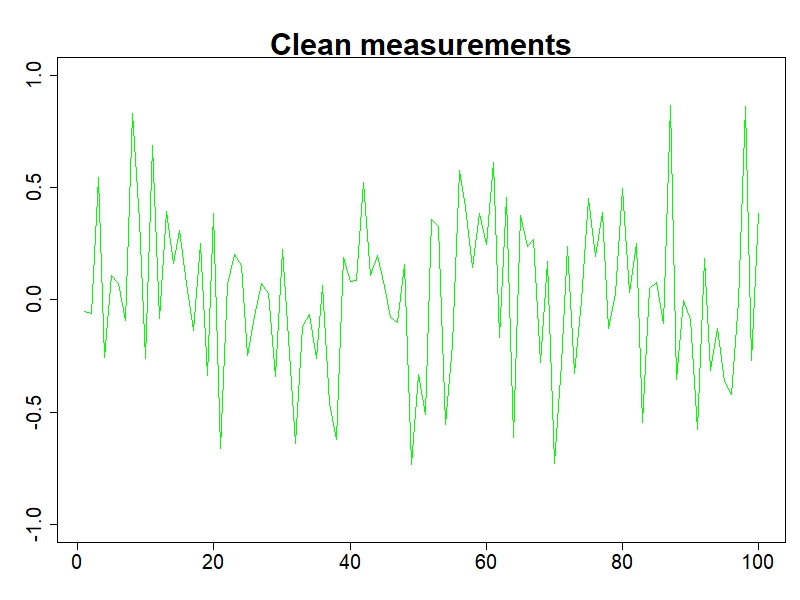}}
		\subfigure[$r$=0.9]{\label{fig:b2}\includegraphics[width=50mm]{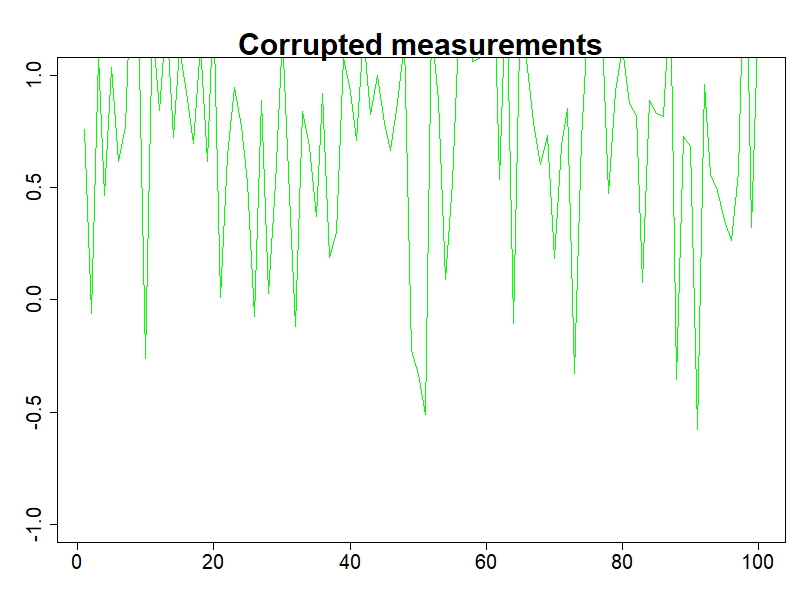}}
		\subfigure[$r$=0.1]{\label{fig:b3}\includegraphics[width=50mm]{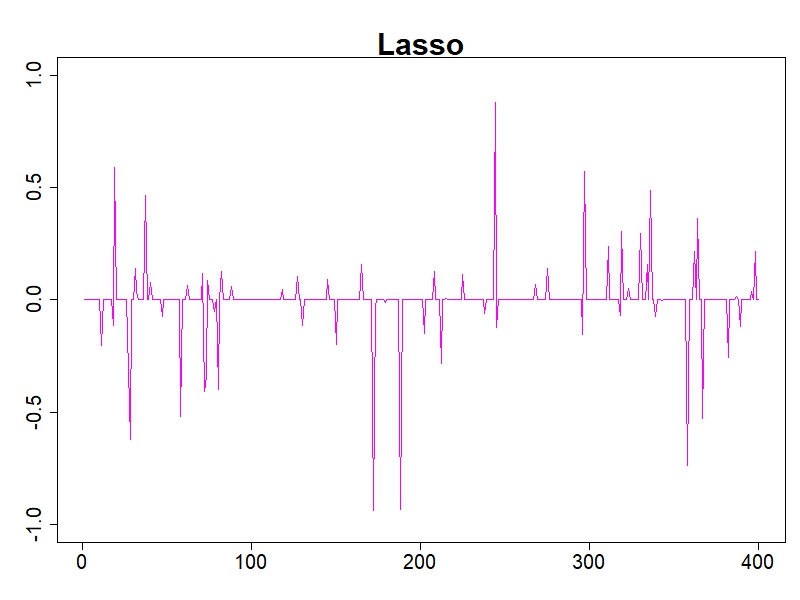}}
		\subfigure[$r$=0.5]{\label{fig:b4}\includegraphics[width=50mm]{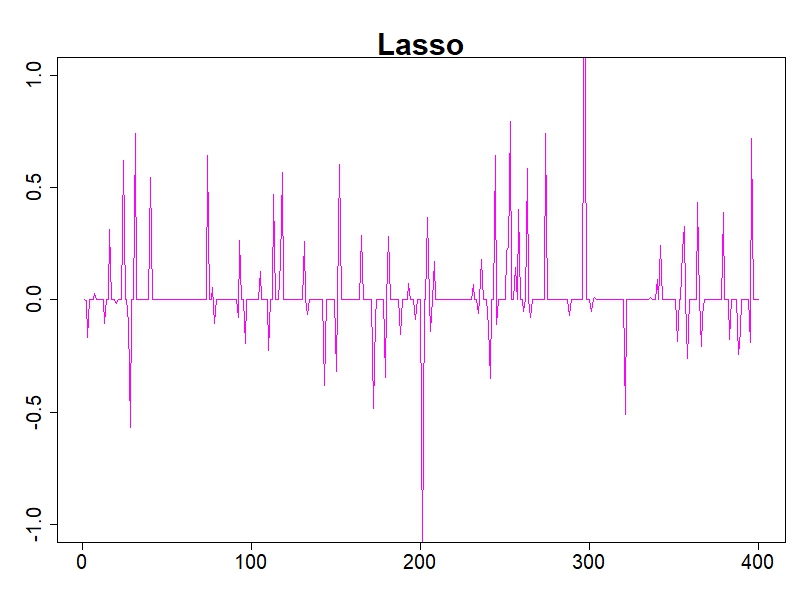}}
		\subfigure[$r$=0.9]{\label{fig:b5}\includegraphics[width=50mm]{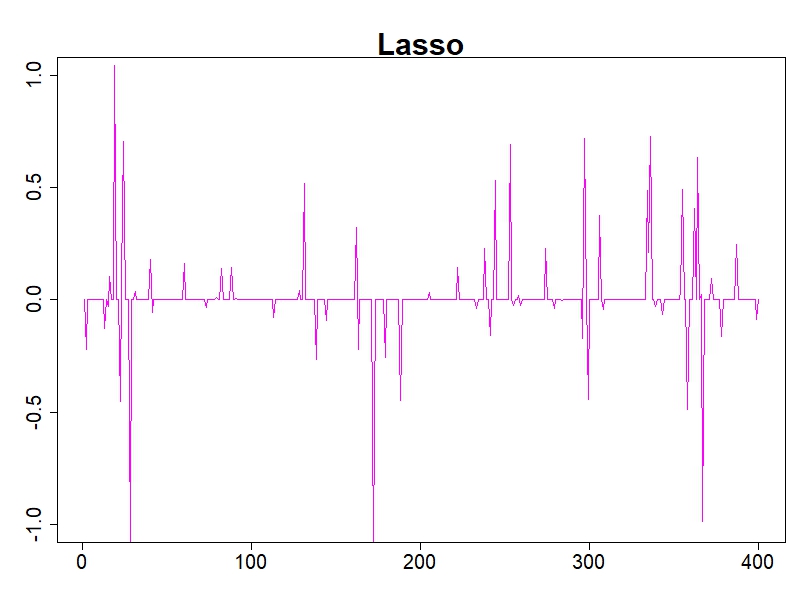}}
		\subfigure[$r$=0.1]{\label{fig:b6}\includegraphics[width=50mm]{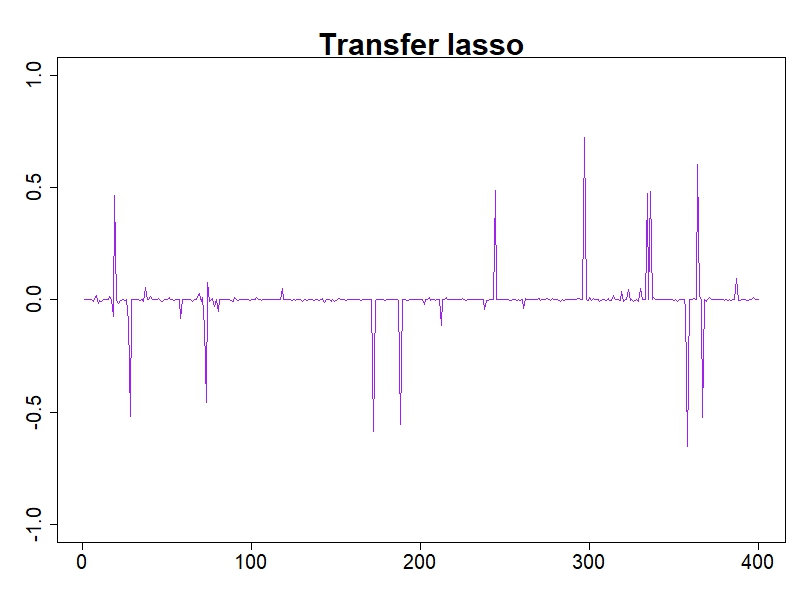}}
		\subfigure[$r$=0.5]{\label{fig:b7}\includegraphics[width=50mm]{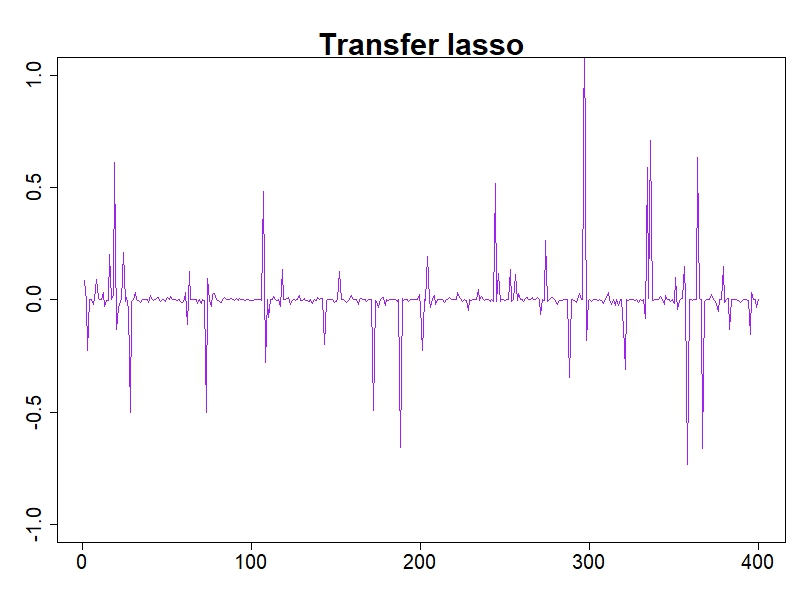}}
		\subfigure[$r$=0.9]{\label{fig:b8}\includegraphics[width=50mm]{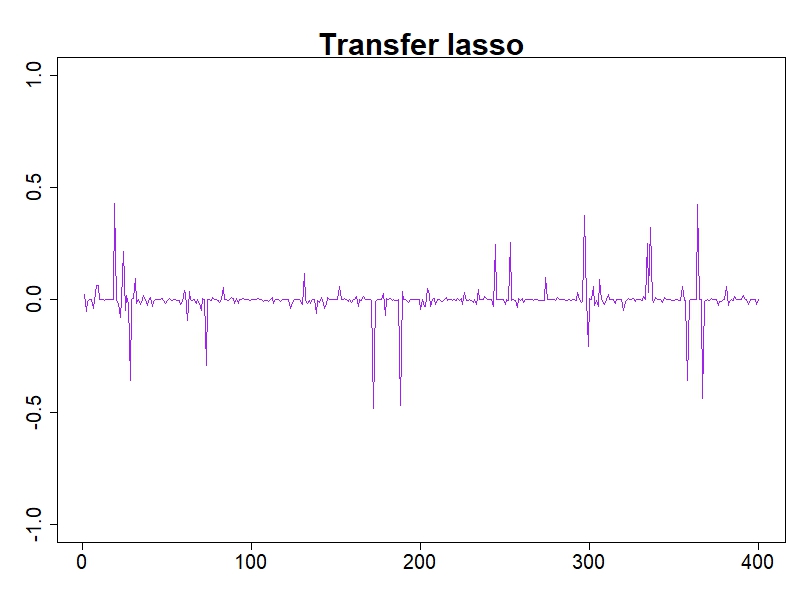}}
		\subfigure[$r$=0.1]{\label{fig:b9}\includegraphics[width=50mm]{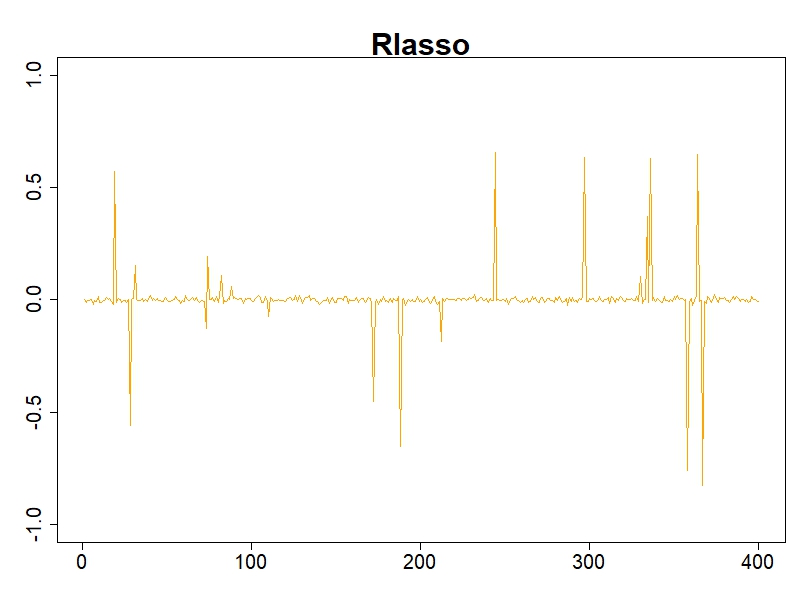}}
		\subfigure[$r$=0.5]{\label{fig:b10}\includegraphics[width=50mm]{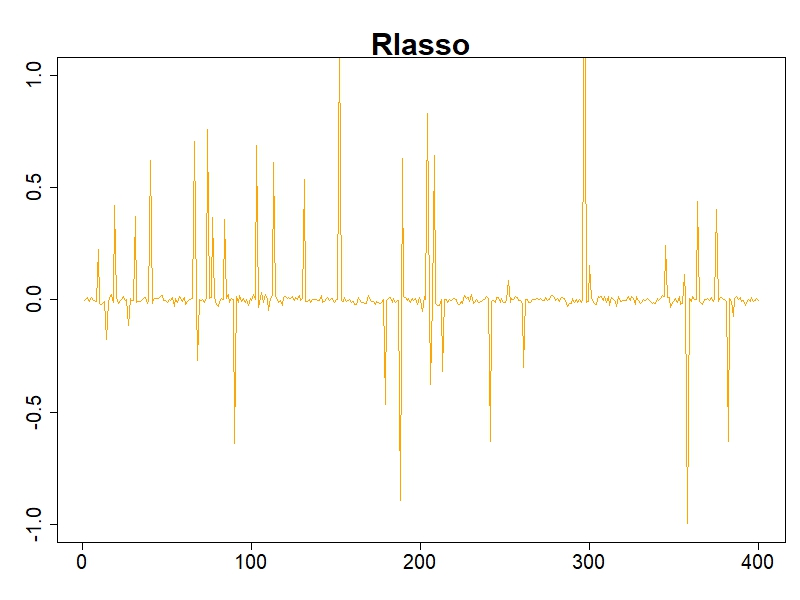}}
		\subfigure[$r$=0.9]{\label{fig:b11}\includegraphics[width=50mm]{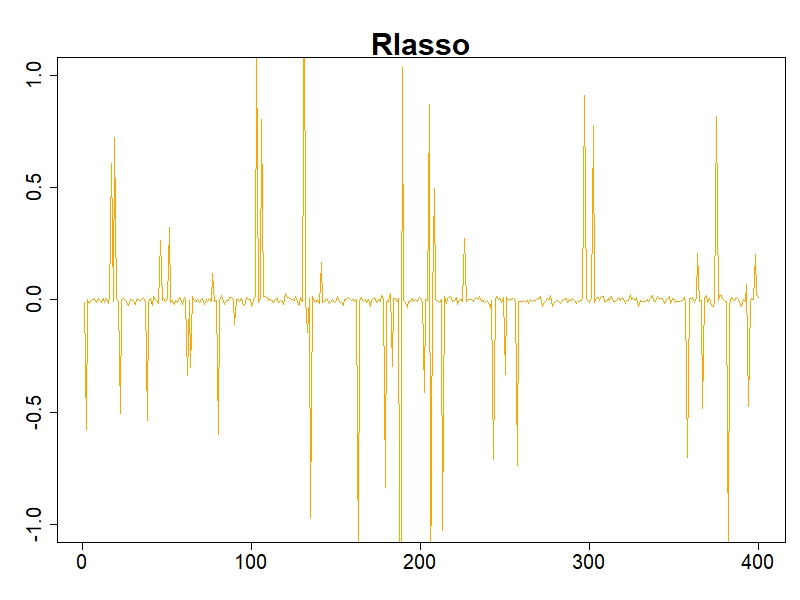}}
		\subfigure[$r$=0.1]{\label{fig:b12}\includegraphics[width=50mm]{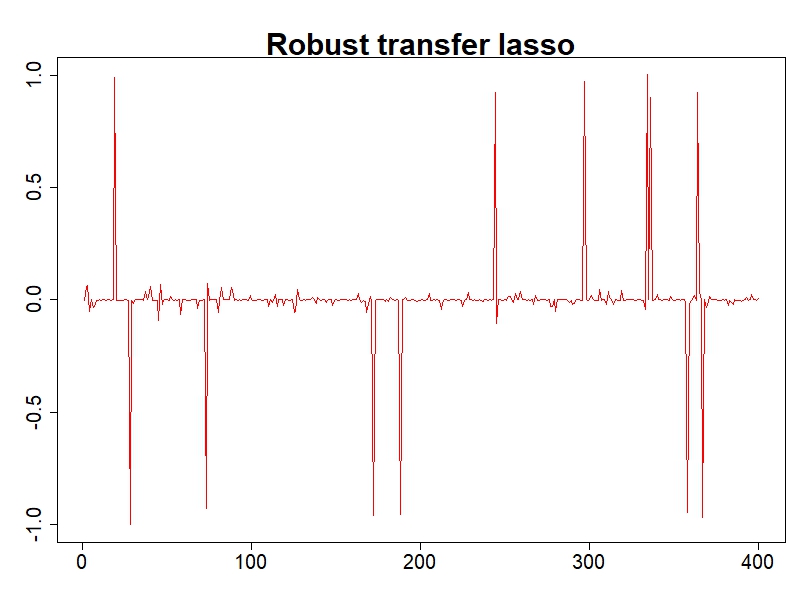}}
		\subfigure[$r$=0.5]{\label{fig:b13}\includegraphics[width=50mm]{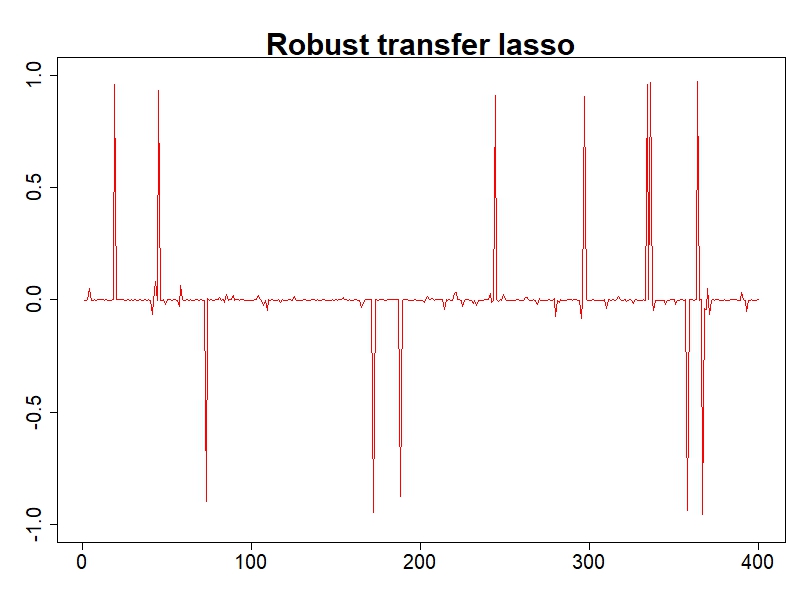}}
		\subfigure[$r$=0.9]{\label{fig:b14}\includegraphics[width=50mm]{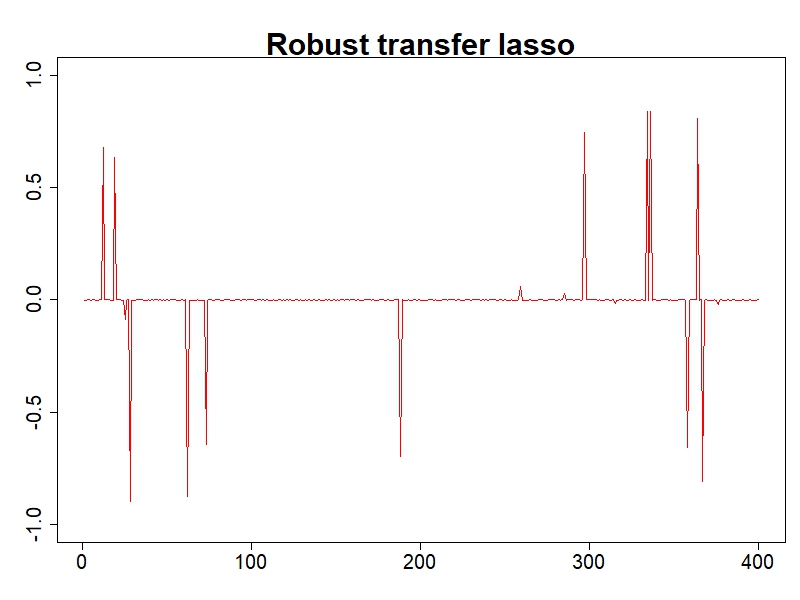}}
		\caption{\label{f1}Sparse signals reconstruction from  corrupted measurements, for $r$ varying from 10\% to 90\%. }
	\end{figure}
	
	\begin{figure}
		\centering
		\includegraphics[width=0.9\linewidth]{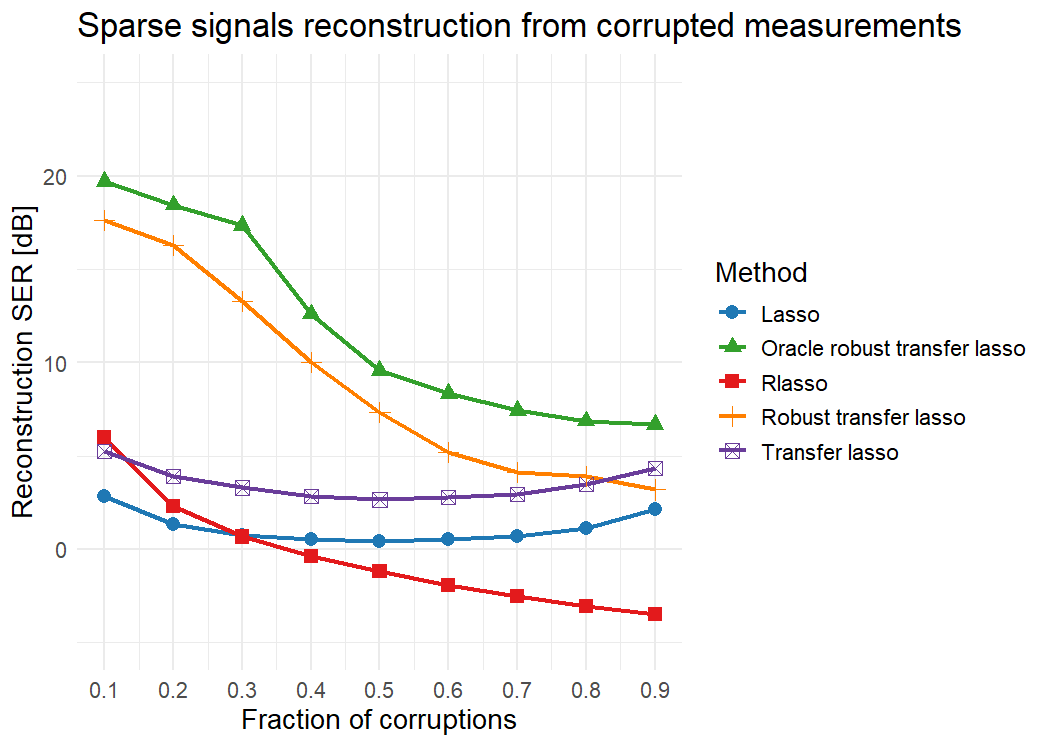}
		\caption{\label{fig:ser}Sparse signals reconstruction from  corrupted measurements, for $r$ varying from 10\% to 90\%. }
	\end{figure}
	
	\section{Analysis of MGMT Methylation and Gene Expression in GBM}
	Glioblastoma(GBM) is a highly aggressive brain tumor with limited treatment options, characterized by rapid proliferation, diffuse infiltration, and resistance to therapy.  O6-methylguanine-DNA methyltransferase (MGMT) methylation plays a pivotal role not only in predicting response to chemotherapy but also in guiding personalized treatment strategies and informing prognosis of GBM. In the field of bioinformatics, numerous studies have explored the relationship between DNA methylation and gene expression, including comprehensive analyses such as those reported by \cite{cancer2012comprehensive}. By employing transfer learning techniques, we aim to analyze the complex interplay between MGMT methylation and gene expression of GBM patients, leveraging the comprehensive datasets available from the Genotype-Tissue Expression (GTEx) project and The Cancer Genome Atlas(TCGA) data. The TCGA Glioblastoma cohort is selected as the target dataset, while the GDC TCGA Glioblastoma cohort is selected as source data 1 and the TCGA Lower Grade Glioma and Glioblastoma cohorts is chosen as source data 2. All datasets can be downloaded from the Xena Browser platform at https://xenabrowser.net/datapages/.
	
	The study begins with differential gene analysis, by leveraging normal brain sample data from GTEx and cancer brain sampel data from TCGA. These repositories provide invaluable resources by offering large-scale genomic and transcriptomic profiles across diverse healthy and pathological tissue samples. Afeter differential gene analysis, the samples without genetic testing data will be excluded from analysis. The data pre-processing flowchart is detailed in Fig. \ref{fig:flowchat}.
	\begin{figure}
		\centering
		\includegraphics[width=0.5\linewidth]{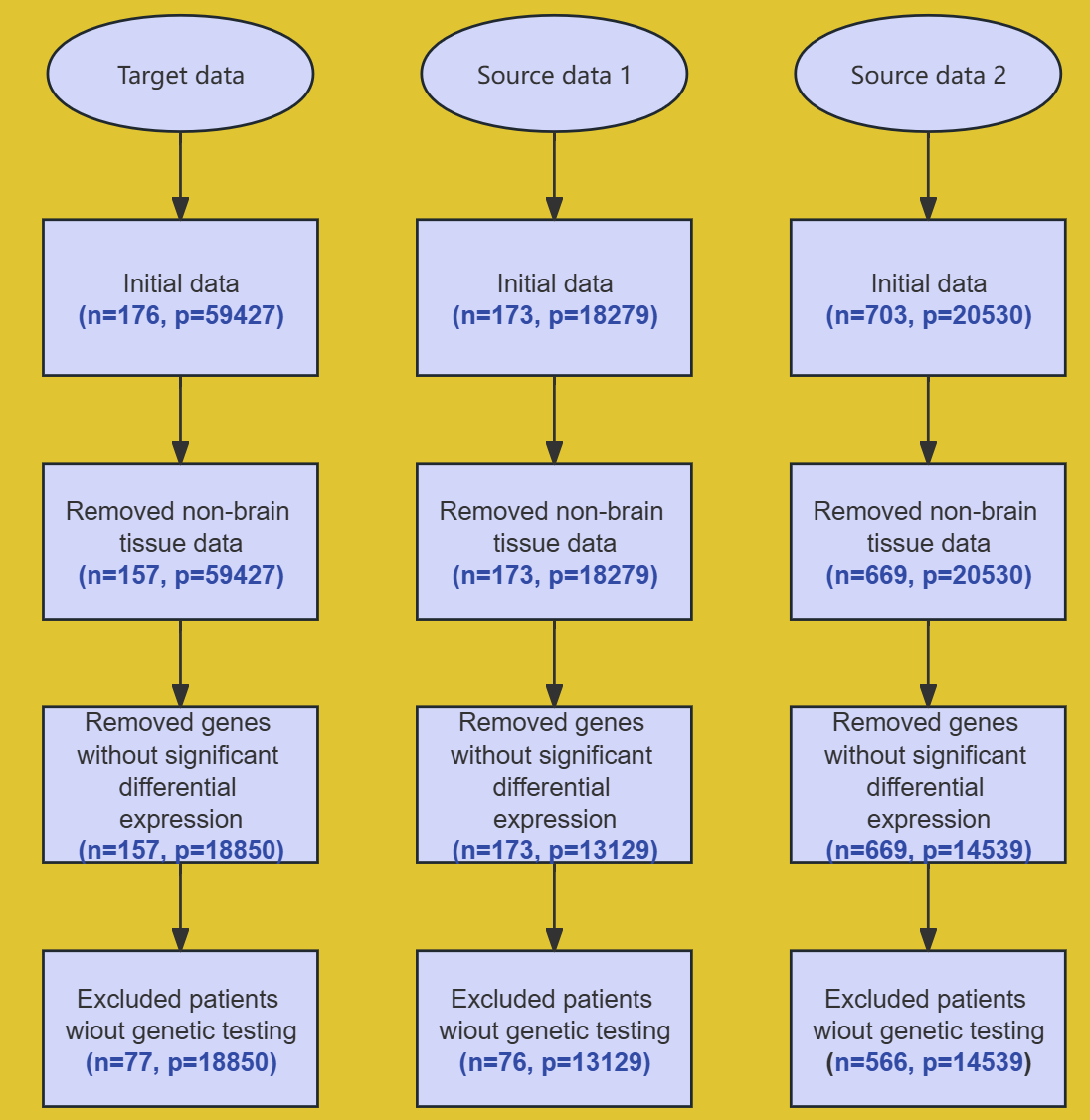}
		\caption{\label{fig:flowchat}Flowchart}
	\end{figure}
	The MGMT methylation data for the target dataset are incomplete, and we treat this incompleteness as corrupted data within our model. In contrast, MGMT methylation data for the two source datasets are nearly complete.
	For feature selection,  two genes, "FBN2" and "SNX31" were selected by both the target dataset and source 1, while no common selected genes were shared between the target dataset and Source 2. The reason maybe source  2 contains also low grade glioma. So only source 1 is selected for transfer regression.

	Unlike in simulations, there is no objective measure for identification accuracy in real-world scenarios. To indirectly address this issue, we perform a prediction evaluation based on 10-fold cross-validation. The prediction MSEs are 0.083(proposed), 0.157(Rlasso), 0.523(Lasso), 0.517(Transfer Lasso).
	
	The selected genes are presented in Figure \ref{gene}. We conducted a Gene Ontology (GO) enrichment analysis on the results obtained from the Robust Transfer Lasso method to identify the biological processes in GBM that are significantly associated with differentially expressed genes. The findings, illustrated in Figure \ref{go}, highlight the top-enriched GO terms, complete with their corresponding p-values and the number of genes involved in each category. Notably, several pathways related to astrocytes were identified, which have been historically linked to GBM. These include processes such as astrocyte development and differentiation (\cite{sofroniew2010astrocytes}, \cite{barres2008mystery}). Other processes, such as the cell surface receptor protein serine/threonine kinase signaling pathway, have also been reported in historical studies related to cancer (\cite{massague2008tgfbeta}).
		\begin{figure}
		\centering
		\includegraphics[width=1\linewidth]{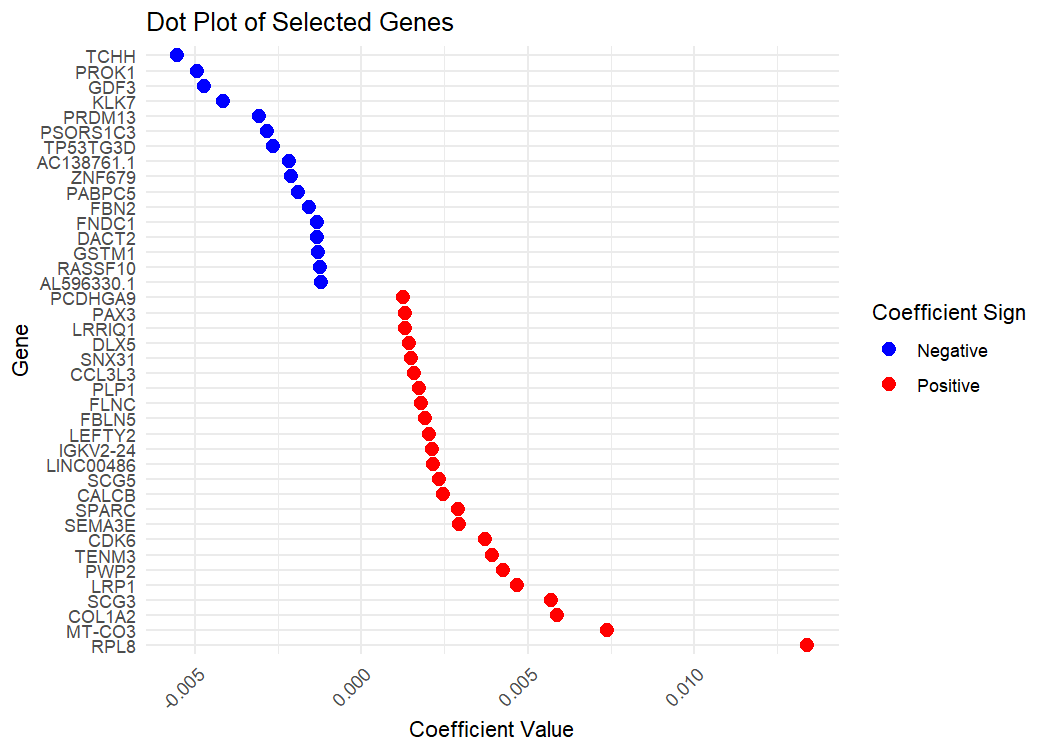}
		\caption{\label{gene}Selected genes}
	\end{figure}
	\begin{figure}
		\centering
		\includegraphics[width=1\linewidth]{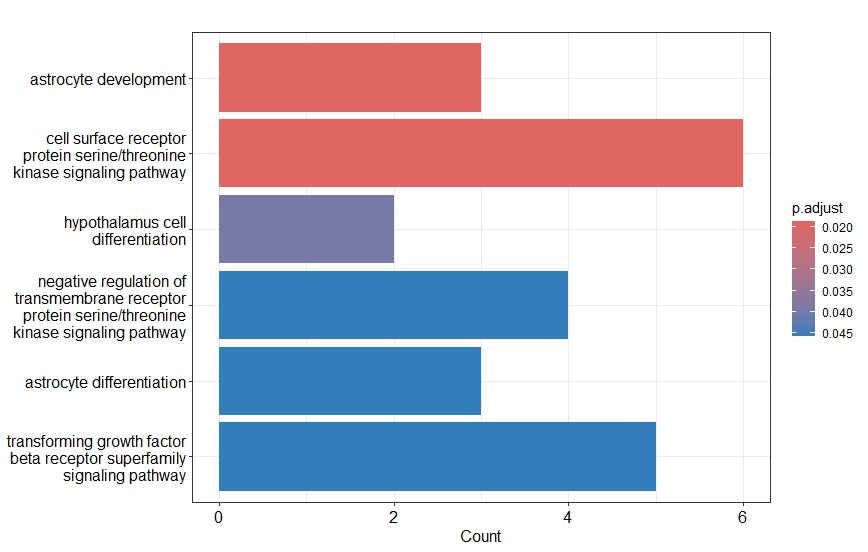}
		\caption{\label{go}GO enrichment analysis}
	\end{figure}

	\section{Discussion\label{sec5}}
	We propose an interpretable transfer learning framework for addressing potential corrupted labels, which can also be extended to handle missing label scenarios. Both theoretical analysis and simulation experiments demonstrate that our method outperforms conventional approaches that rely solely on target data, highlighting its robustness and effectiveness in practical applications.
	
	Several promising directions exist for extending this work in the future. First, in real-world datasets, corruption often affects not only labels but also predictors, which necessitates the development of methods to simultaneously address both types of corruption. While our current framework assumes that all source data are "clean," this assumption may not hold in certain real-world applications. To address this limitation, outlier detection techniques could be integrated into the preprocessing phase to identify and mitigate potential contamination in the source data before applying transfer regression. Additionally, future research could explore adaptive weighting schemes to dynamically balance the contributions of source and target data, particularly in scenarios where the quality of source data varies significantly. Finally, extending the proposed framework to handle non-linear relationships could further broaden its applicability to complex real-world problems.
	
	\section*{Acknowledgements}
	The authors would like to thank the anonymous referees, an Associate Editor and the Editor for their constructive comments that improved the quality of this paper.

	\section*{Appendix: Proof of Main Results
	}
Denote 
\begin{equation}
	\widetilde{\mathcal{L}}_v(\beta, \widehat{\ubeta}^{(t)}) = \mathcal{L}_v(\beta) + \left\langle \frac{1}{|\mathcal{A}_h|} \sum_{j \in \mathcal{A}_h} \nabla \mathcal{L}_j(\widehat{\ubeta}^{(t)}) - \nabla \mathcal{L}_{v}(\widehat{\ubeta}^{(t)}), \beta \right\rangle,
\end{equation}
where $\widehat{\ubeta}^{(t)}$ represents the estimated parameter vector at the $t$-th iteration for $\hat{\ubeta}^{\mathcal{A}_h}$. The following lemma extends theoretical results from distributed Lasso to scenarios where the coefficients of each dataset may differ, directly leveraging results from \cite{wang2017efficient}.

\begin{lemma}\label{lem:dl}
	Under conditions C1-C2, for sufficiently large $t$, with probability at least $1-2x$, we have
	\begin{align*}
		\left\|\widehat{\ubeta}^{(t)} - \bar{\ubeta}^{\mathcal{A}_h}\right\|_1 &\leq 49 \bar{s} \sigma_{\epsilon} \sqrt{\frac{\log(p/x)}{|\mathcal{A}_h|n}},\\
			\left\|\widehat{\ubeta}^{(t)} - \bar{\ubeta}^{\mathcal{A}_h}\right\|_2 &\leq 13 \bar{s} \sigma_{\epsilon} \sqrt{\frac{\log(p/x)}{|\mathcal{A}_h|n}}.
	\end{align*}
	Furthermore, with probability approaching 1, it holds that
	\begin{align*}
		&\left\| \frac{1}{|\mathcal{A}_h|} \sum_{j \in \mathcal{A}_h} \nabla \mathcal{L}_j(\widehat{\ubeta}^{(t)}) - \nabla \mathcal{L}_{v}(\widehat{\ubeta}^{(t)}) \right\|_{\infty} \\
		&\quad\leq \left\| \frac{1}{|\mathcal{A}_h|} \sum_{j \in \mathcal{A}_h} \nabla \mathcal{L}_j(\bar{\ubeta}^{\mathcal{A}_h}) \right\|_{\infty} + \frac{\log(np)}{n} \sqrt{\frac{\log(2p)}{n}} \| \bar{\ubeta}^{\mathcal{A}_h} - \widehat{\ubeta}^{(t)} \|_1 \\
		&\quad\quad+ C\left(\frac{\log(np)}{n}\right)^{2/3} \| \bar{\ubeta}^{\mathcal{A}_h} - \widehat{\ubeta}^{(t)} \|_1^2,
	\end{align*}
	for some universal constant $C$.
\end{lemma}

\begin{proof}
	By applying Theorem 6 of \cite{wang2017efficient}, we obtain that with probability at least $1-2x$,
	\begin{align*}
		\left\|\widehat{\ubeta}^{(t)} - \bar{\ubeta}^{\mathcal{A}_h}\right\|_1 &\leq \frac{48 \bar{s} \sigma_{\epsilon} \|\mathrm{diag}(\Sigma)\|_{\infty}}{\|\hat{\Sigma}\|_{\infty}} \sqrt{\frac{\log(p/x)}{|\mathcal{A}_h|n}} + C\left(\sqrt{\frac{\log(2p/x)}{n}}\right)^{t}\|\widehat{\boldsymbol{\ubeta}}_0 - \bar{\ubeta}^{\mathcal{A}_h}\|_1,\\
		\left\|\widehat{\ubeta}^{(t)} - \bar{\ubeta}^{\mathcal{A}_h}\right\|_2 &\leq 12 \bar{s} \sigma_{\epsilon} \sqrt{\frac{\log(p/x)}{|\mathcal{A}_h|n}} + C\left(\sqrt{\frac{\log(2p/x)}{n}}\right)^{t}\|\widehat{\boldsymbol{\ubeta}}_0 - \bar{\ubeta}^{\mathcal{A}_h}\|_2
	\end{align*}
	which implies the first inequality for large enough $t$. 		By Theorem 9.3 in \cite{fan2020statistical},
	\begin{align*}
		\|\hat{\Sigma}-\Sigma\|_{\infty} = O_P\left(\sqrt{\frac{\log p}{n}}\right),
	\end{align*}
	which completes the proof of first part. The second part is a direct consequence of Lemma 8 in \cite{wang2017efficient}.
\end{proof}
	
	\begin{lemma}\label{Lemma:1}
		Assume conditions C1-C4 hold. Then, for sufficiently large $t$,
		\begin{align}\label{lem1}
			\frac{1}{n}\|\mathbb{X}(\hat{\ubeta}^{\mathcal{A}_h}-\bar{\ubeta}^{\mathcal{A}_h})\|_{2}^2=O_P\left(\frac{\bar{s}\log(p)}{|\mathcal{A}_h|n}\right),
		\end{align}
		where $\bar{s} = \|\bar{\ubeta}^{\mathcal{A}_h}\|_0$.
	\end{lemma}
	
	\begin{proof}
		By Theorem 9.3 in \cite{fan2020statistical},
		\begin{align}\label{lem1:eq1}
			\|\hat{\Sigma}-\Sigma\|_{\infty} = O_P\left(\sqrt{\frac{\log p}{n}}\right).
		\end{align}
		Denote $\mathcal{Z}=\hat{\ubeta}^{\mathcal{A}_h}-\bar{\ubeta}^{\mathcal{A}_h}$, under condition C1, by Lemma \ref{lem:dl} and (\ref{lem1:eq1}), it holds that
		\begin{align}\label{lem1:seq1}
			\begin{split}
			\frac{1}{n}\|\mathbb{X}(\hat{\ubeta}^{\mathcal{A}_h}-\bar{\ubeta}^{\mathcal{A}_h})\|_{2}^2&=\mathcal{Z}\Sigma\mathcal{Z}+\mathcal{Z}\l(\hat{\Sigma}-\Sigma\r)\mathcal{Z}\\
			&\leq \lambda_{\max}(\Sigma)\|\mathcal{Z}\|_2^2+\l\|\hat{\Sigma}-\Sigma\r\|_{\infty}\|\mathcal{Z}\|_1^2\\
			&=O_P\l(\bar{s}\frac{\log(p)}{|\mathcal{A}_h|n} +\bar{s}^2\frac{\log(p)}{|\mathcal{A}_h|n}\sqrt{\frac{\log p}{n}} \r).
			\end{split}
		\end{align}
		By the defination of $\mathcal{A}_h$, it holds that under condition C4 
		$$
		\|\Delta^{(\mathrm{S}_j)}\|_0\leq \gamma_1^{-1}h,\quad j\in \mathcal{A}_h.
		$$
		Thus $\bar{s} \leq \gamma_1^{-1}|\mathcal{A}_h|h$. By condition C3,
		we have that
		$$
		\bar{s}^2\frac{\log(p)}{|\mathcal{A}_h|n}\sqrt{\frac{\log p}{n}}=o\l(\bar{s}\frac{\log(p)}{|\mathcal{A}_h|n} \r).
		$$
		By (\ref{lem1:seq1}), the proof is completed.
	\end{proof}
	\begin{lemma}\label{Lemma:2}
		Assume conditions C1-C2 hold. Further suppose that 
		$$
		\frac{1}{n} \|\mathbb{X}\epsilon\|_\infty \leq  \frac{\lambda_{\Delta}}{2}  \quad \text{and} \quad \frac{1}{\sqrt{n}} \|\epsilon\|_\infty \leq \frac{\lambda_{e}}{2},
		$$
		then
		\begin{align}\label{lem2}
			\|\hat{\Delta}^{\mathcal{A}_h}-\Delta^{\mathcal{A}_h}\|_{2}+ \|\hat{e}^{\mathcal{A}_h}-e^*\|_{2}=O_P\left(\sqrt{\frac{\bar{s}\log p}{|\mathcal{A}_h|n}}+ \sqrt{\frac{s_{\Delta}\log p}{n}}+\sqrt{\frac{k\log n}{n}}\right).
		\end{align}
	\end{lemma}
	
\begin{proof}
	By the definition of $\hat{\Delta}^{\mathcal{A}_h}$ and $\hat{e}^{\mathcal{A}_h}$,
	\begin{equation}\label{eq11.19}
		\begin{aligned}
			&\frac{1}{2n} \left\| \mathbb{Y}- \mathbb{X} \hat{\ubeta}^{\mathcal{A}_h} - \mathbb{X} \hat{\Delta}^{\mathcal{A}_h} - \sqrt{n} \hat{e}^{\mathcal{A}_h} \right\|_2^2 \\
			&\quad+ \lambda_{\Delta} \left\| \hat{\Delta}^{\mathcal{A}_h} \right\|_1 + \lambda_{e} \left\| \hat{e}^{\mathcal{A}_h} \right\|_1 \\
			&\leq \frac{1}{2n} \left\|\mathbb{Y} - \mathbb{X} \hat{\ubeta}^{\mathcal{A}_h} - \mathbb{X} \Delta^{\mathcal{A}_h} - \sqrt{n} e^* \right\|_2^2 \\
			&\quad+ \lambda_{\Delta} \left\| \Delta^{\mathcal{A}_h} \right\|_1 + \lambda_{e} \left\| e^* \right\|_1.
		\end{aligned}
	\end{equation}
	
	Define $z=\hat{\Delta}^{\mathcal{A}_h}-\Delta^{\mathcal{A}_h}$ and $v=\hat{e}^{\mathcal{A}_h}-e^*$, it holds that:
	\begin{equation}\label{eq11.20}
		\begin{aligned}
			\left\| \mathbb{Y}- \mathbb{X} \hat{\ubeta}^{\mathcal{A}_h} - \mathbb{X} \hat{\Delta}^{\mathcal{A}_h} - \sqrt{n} e^* \right\|_2^2 =&
			\left\| \mathbb{Y}- \mathbb{X} \hat{\ubeta}^{\mathcal{A}_h} -  \mathbb{X} \Delta^{\mathcal{A}_h} -\sqrt{n} e^* \right\|_2^2+ \left\|  \mathbb{X}z + \sqrt{n} v \right\|_2^2 \\
			&- 2 \langle  \mathbb{Y}- \mathbb{X}\hat{\ubeta}^{\mathcal{A}_h} -  \mathbb{X} \Delta^{\mathcal{A}_h} -\sqrt{n} e^*,  \mathbb{X} z+ \sqrt{n}v  \rangle .
		\end{aligned}
	\end{equation}
	
	Moreover,
	\begin{align*}
		\left\|\Delta^{\mathcal{A}_h}\right\|_1 - \left\|\widehat{\Delta}\right\|_1 &= \left\|\Delta^{\mathcal{A}_h}\right\|_1 - \left\|\Delta^{\mathcal{A}_h} + z\right\|_1 \\
		&= \left\|\Delta^{\mathcal{A}_h}\right\|_1 - \left\|\Delta^{\mathcal{A}_h} + z_T\right\|_1 - \left\|z_{T^c}\right\|_1 \\
		&\leq \left\|z_T\right\|_1 - \left\|z_{T^c}\right\|_1,
	\end{align*}
	similarly,
	\begin{align*}
		\left\|e\right\|_1 - \left\|\widehat{e}\right\|_1 \leq \left\|v_E\right\|_1 - \left\|v_{E^c}\right\|_1.
	\end{align*}
	
	Combining these pieces together yields:
	\vspace{-10pt}
	\begin{equation}
		\begin{aligned}
			\frac{1}{2n} \|\mathbb{X}z + \sqrt{n} v\|_2^2 
			&\leq \frac{1}{n} \langle \mathbb{Y} - \mathbb{X} \hat{\ubeta}^{\mathcal{A}_h} - \mathbb{X} \Delta^{\mathcal{A}_h} - \sqrt{n} e^*, \mathbb{X}z + \sqrt{n} v \rangle \\
			&\quad+ \lambda_{\Delta} (\|z_T\|_1 - \|z_{T^c}\|_1) + \lambda_{e} (\|v_E\|_1 - \|v_{E^c}\|_1) \\
			&\leq \frac{1}{n} \|\mathbb{X}^\top \epsilon\|_\infty \|z\|_1 + \frac{1}{\sqrt{n}} \|\epsilon\|_\infty \|v\|_1 \\
			&\quad+ \frac{1}{4n} \|\mathbb{X}z + \sqrt{n} v\|_2^2 + \frac{1}{n} \|\mathbb{X}(\bar{\ubeta}^{\mathcal{A}_h}-\hat{\ubeta}^{\mathcal{A}_h})\|_2^2 \\
			&\quad+ \lambda_{\Delta} (\|z_T\|_1 - \|z_{T^c}\|_1) + \lambda_{e} (\|v_E\|_1 - \|v_{E^c}\|_1) \\
			&\leq \left(\frac{1}{n} \|\mathbb{X}^\top \epsilon\|_\infty + \lambda_{\Delta}\right) \|z_T\|_1 - \left(\lambda_{\Delta} - \frac{1}{n} \|\mathbb{X}^\top \epsilon\|_\infty \right) \|z_{T^c}\|_1 \\
			&\quad+ \frac{1}{4n} \|\mathbb{X}z + \sqrt{n} v\|_2^2 + \frac{1}{n} \|\mathbb{X}(\bar{\ubeta}^{\mathcal{A}_h}-\hat{\ubeta}^{\mathcal{A}_h})\|_2^2 \\
			&\quad+ \left(\frac{1}{\sqrt{n}} \|\epsilon\|_\infty  + \lambda_{e}\right) \|v_E\|_1 - \left(\lambda_{e} - \frac{1}{\sqrt{n}} \|\epsilon\|_\infty \right) \|v_{E^c}\|_1.
		\end{aligned}
	\end{equation}
	
	By the choice of $\lambda_{\Delta}$ and $\lambda_{e}$ and Lemma \ref{Lemma:1}, with probability approaching 1 as $n \to \infty$, it holds that,
	\begin{align}\label{l2}
		\begin{split}
			\frac{1}{4n} \|\mathbb{X}z + \sqrt{n} v\|_2^2 &\leq \frac{3}{2}\lambda_{\Delta}\|z_T\|_1-
			\frac{1}{2}\lambda_{\Delta} \|z_{T^c}\|_1+\frac{3}{2}\lambda_{e}\|v_E\|_1- \frac{1}{2}\lambda_{e} \|v_{E^c}\|_1\\
			&\quad + \frac{1}{n} \|\mathbb{X}(\bar{\ubeta}^{\mathcal{A}_h}-\hat{\ubeta}^{\mathcal{A}_h})\|_2^2\\
			&\leq \frac{3}{2}\lambda_{\Delta}\|z_T\|_1-
			\frac{1}{2}\lambda_{\Delta} \|z_{T^c}\|_1+\frac{3}{2}\lambda_{e}\|v_E\|_1- \frac{1}{2}\lambda_{e} \|v_{E^c}\|_1\\
			&\quad + C\frac{\bar{s}\log p}{|\mathcal{A}_h|n},
		\end{split}   
	\end{align}
	for some universal constant $C$. Then we have that with probability approaching 1 as $n \to \infty$,
	\begin{align}\label{rec0}
		\begin{split}
			\lambda_{\Delta} \|z\|_1&\leq  4\lambda_{\Delta}\|z_T\|_1 + 3\lambda_{e}\|v_E\|_1
			+  C\frac{\bar{s}\log p}{|\mathcal{A}_h|n}\\
			&\leq  4\sqrt{s_{\Delta}}\lambda_{\Delta}\|z\|_2 + 3\sqrt{k}\lambda_{e}\|v_E\|_2
			+  C\frac{\bar{s}\log p}{|\mathcal{A}_h|n}.
		\end{split}
	\end{align}
	
	By the extend RE condition(C1), 
	$$
	\frac{1}{4n} \|\mathbb{X}z + \sqrt{n} v\|_2^2 \geq \kappa_l^2(\| z \|_2 + \| f \|_2)^2 +  C\frac{\bar{s}\log p}{|\mathcal{A}_h|n}\sqrt{\frac{\log p}{n}}.
	$$
	
	Hence, with probability approaching 1 as $n \to \infty$,
	\begin{align*}
		\kappa_l^2(\|z\|_2 + \|v\|_2)^2 &\leq 4\lambda_{\Delta}\|z_T\|_1 + 4\lambda_{e}\|v_S\|_1 +  C\frac{\bar{s}\log p}{|\mathcal{A}_h|n}\\
		&\leq 4\lambda_{\Delta}\sqrt{s_{\Delta}}\|z\|_2 + 4\lambda_{e}\sqrt{k}\|v\|_2+  C\frac{\bar{s}\log p}{|\mathcal{A}_h|n}.
	\end{align*}
	Thus
	$$
	\|z\|_2 + \|v\|_2 \leq 4\kappa_l^{-2}\l[\max\l\{\lambda_{\Delta}\sqrt{s_{\Delta}}, \lambda_{e}\sqrt{k}\r\}+C\sqrt{\frac{\bar{s}\log p}{|\mathcal{A}_h|n}}\r].
	$$
		Notice that under condition C1-C2, it holds with probability to 1 that,
	\begin{align*}
	\frac{\|\mathbb{X}\epsilon\|_{\infty}}{n} &\leq 2\sigma_{\epsilon}K_{clm}\sqrt{\frac{\log p}{n}},\\
		\frac{\|\epsilon\|_{\infty}}{\sqrt{n}} &\leq 2\sigma_{\epsilon}\sqrt{\frac{\log n}{n}},
	\end{align*}
	which completes our proof.
\end{proof}
		\begin{lemma}\label{Lemma:5}
		Assume conditions C1-C4 hold, then for large enough $t$, it holds with probability approach 1 as $n \to \infty$,
		$$
		\l\|\hat{\ubeta}^{\mathcal{A}_h} -\bar{\ubeta}^{\mathcal{A}_h} \r\|_{\infty}\leq 3\sigma_{\epsilon}K_{clm}\sqrt{\frac{\log p}{n}}+4\frac{1}{\sqrt{\bar{C}_{\min}}}\sigma_{\epsilon}\lambda_{t}+\l\|\l(\frac{\mathbb{X}_{\bar{T}_h}^{\mathrm{S}_v\top} \mathbb{X}_{\bar{T}_h}^{\mathrm{S}_v}}{n}\r)^{-1}\mathrm{sign}(\bar{\ubeta}^{(t)}_{(\bar{T}_h)})\r\|_{\infty}\lambda_{t},
		$$
		where $\hat{\Sigma}^{\mathrm{S}_v}=\mathbb{X}^{\mathrm{S}_v\top}\mathbb{X}^{\mathrm{S}_v}/n$, and 
		$$
		\bar{C}_{\min}:=\lambda_{\min}\left(\frac{\mathbb{X}_{\bar{T}_h}^{\mathrm{S}_v\top} \mathbb{X}_{\bar{T}_h}^{\mathrm{S}_v}}{n}\right).
		$$
	\end{lemma}
		\begin{proof}
	By the zero-subgradient conditions, at the $t$-th iteration, it holds that
	\begin{align}\label{lem5:kkt}
		\begin{split}
			-\frac{1}{n}\mathbb{X}^{\mathrm{S}_v\top}(\mathbb{Y}^{\mathrm{S}_v}-\mathbb{X}^{\mathrm{S}_v}\bar{\ubeta}^{\mathcal{A}_h})+\frac{1}{|\mathcal{A}_h|} \sum_{j \in \mathcal{A}_h} \nabla \mathcal{L}_j(\widehat{\ubeta}^{(t)}) - \nabla \mathcal{L}_{v}(\widehat{\ubeta}^{(t)})+\lambda_t\bar{z}=0,
		\end{split}
	\end{align}
	where 
	$$
	\frac{1}{|\mathcal{A}_h|} \sum_{j \in \mathcal{A}_h} \nabla \mathcal{L}_j(\widehat{\ubeta}^{(t)}) - \nabla \mathcal{L}_{v}(\widehat{\ubeta}^{(t)})=-\frac{1}{|\mathcal{A}_h|}\sum_{j \in \mathcal{A}_h}\frac{1}{n}\mathbb{X}^{\mathrm{S}_j\top}(\mathbb{Y}^{(\mathrm{S}_j)}-\mathbb{X}^{(\mathrm{S}_j)}\hat{\ubeta}^{(t)})+
	\frac{1}{n}\mathbb{X}^{\mathrm{S}_v\top}(\mathbb{Y}^{\mathrm{S}_v}-\mathbb{X}^{\mathrm{S}_v}\hat{\ubeta}^{(t)}),
	$$
	and $ \bar{z} \in \partial \|\ubeta\|_1 $ is a sub gradient. Denote $\bar{T}_h$ as the support set of $\bar{\ubeta}^{\mathcal{A}_h}$.
	We use the primal-dual witness method(\cite{mj2009sharp}) :\\
	(i): Set $\widehat{\ubeta}^{(t)}_{(\bar{T}_h^c)}=0$.\\
	(ii): Determine $\widehat{\ubeta}^{(t)}_{(\bar{T}_h)}, \bar{z}$ by solving (\ref{lem5:kkt}).\\
	(iii):Check whether or not the strict dual feasibility condition $\|\bar{z}_{(\bar{T}_h^c)}\|_{\infty}<1$ hold.\\
	Writing the zero-subgradient conditions $\ref{lem5:kkt}$ in  block matrix form, we obtain
	$$
	\frac{1}{n}\mathcal{D}
	\begin{bmatrix}
		\widehat{\ubeta}^{(t)}-\bar{\ubeta}^{\mathcal{A}_h}_{(T)} \\
		0
	\end{bmatrix}
	- \frac{1}{n}
	\begin{bmatrix}
		\mathbb{X}^{\mathrm{S}_v\top}_{\bar{T}_h} \epsilon^{\mathrm{S}_v}+\l(\frac{1}{|\mathcal{A}_h|} \sum_{j \in \mathcal{A}_h} \nabla \mathcal{L}_j(\widehat{\ubeta}^{(t)}) - \nabla \mathcal{L}_{v}(\widehat{\ubeta}^{(t)} \r)\\
		\mathbb{X}^{\mathrm{S}_v\top}_{\bar{T^c}} \epsilon^{\mathrm{S}_v}+\l(\frac{1}{|\mathcal{A}_h|} \sum_{j \in \mathcal{A}_h} \nabla \mathcal{L}_j(\widehat{\ubeta}^{(t)}) - \nabla \mathcal{L}_{v}(\widehat{\ubeta}^{(t)} \r)
	\end{bmatrix}
	+ \lambda_{t}
	\begin{bmatrix}
		\bar{z}_{(\bar{T}_h)} \\
		\bar{z}_{(\bar{T}_h^c)}
	\end{bmatrix}
	=
	\begin{bmatrix}
		0 \\
		0
	\end{bmatrix}.
	$$
	where
	$$
	\mathcal{D} = \begin{bmatrix}
		\mathbb{X}_{\bar{T}_h}^{\mathrm{S}_v\top} \mathbb{X}_{\bar{T}_h}^{\mathrm{S}_v} & \mathbb{X}_{\bar{T}_h}^{\mathrm{S}_v\top} \mathbb{X}_{\bar{T}_h^c}^{\mathrm{S}_v} \\
		\mathbb{X}_{\bar{T}_h}^{\mathrm{S}_v\top} \mathbb{X}_{\bar{T}_h^c}^{\mathrm{S}_v} & \mathbb{X}_{\bar{T}_h^c}^{\mathrm{S}_v\top} \mathbb{X}_{\bar{T}_h^c}^{\mathrm{S}_v}
	\end{bmatrix}.
	$$
	By \cite{Tibshirani2012TheLP}, under absolutely continuous distribution condition(C1), the solution for $\widehat{\ubeta}^{(t)}$ is unique. Under absolutely continuous distribution condition(C1) and the condition that $\Sigma$ has minimum eigenvalue bounded from 0(C1), $\mathbb{X}_{\bar{T}_h}^{\mathrm{S}_v\top} \mathbb{X}_{\bar{T}_h}^{\mathrm{S}_v}$ is invertible, and satisfy
	$$
	\bar{C}_{\min}:=\lambda_{\min}\left(\frac{\mathbb{X}_{\bar{T}_h}^{\mathrm{S}_v\top} \mathbb{X}_{\bar{T}_h}^{\mathrm{S}_v}}{n}\right)  >0,
	$$
	Solve for the vector $ \widehat{\ubeta}^{(t)}_{(T)}-\bar{\ubeta}^{(t)}_{(T)}$ yields,
	\begin{align}\label{l5:eq1}
		\begin{split}
			\widehat{\ubeta}^{(t)}_{(T)}-\bar{\ubeta}^{(t)}_{(T)}=&\l(\frac{\mathbb{X}_{\bar{T}_h}^{\mathrm{S}_v\top} \mathbb{X}_{\bar{T}_h}^{\mathrm{S}_v}}{n}\r)^{-1}\frac{\mathbb{X}_{\bar{T}_h}^{\mathrm{S}_v\top} \epsilon^{\mathrm{S}_v} }{n}-\lambda_{t}\l(\frac{\mathbb{X}_{\bar{T}_h}^{\mathrm{S}_v\top} \mathbb{X}_{\bar{T}_h}^{\mathrm{S}_v}}{n}\r)^{-1}\mathrm{sign}(\bar{\ubeta}^{(t)}_{(T)})\\
			&+\l(\frac{\mathbb{X}_{\bar{T}_h}^{\mathrm{S}_v\top} \mathbb{X}_{\bar{T}_h}^{\mathrm{S}_v}}{n}\r)^{-1}\frac{1}{n}\l(\frac{1}{|\mathcal{A}_h|} \sum_{j \in \mathcal{A}_h} \nabla \mathcal{L}_j(\widehat{\ubeta}^{(t)}) - \nabla \mathcal{L}_{v}(\widehat{\ubeta}^{(t)})\r).
		\end{split}
	\end{align}
	Solve for the vector $ \bar{z}_{(T^c)}$  yields,
	\begin{align}\label{l5:eq2}
		\begin{split}
			\bar{z}_{(T^c)} &= \frac{1}{\lambda_{t}}\Bigg(\frac{1}{n}\mathbb{X}_{\bar{T}_h^c}^{\mathrm{S}_v\top} \epsilon^{\mathrm{S}_v} 
			- \frac{\mathbb{X}_{\bar{T}_h^c}^{\mathrm{S}_v\top} \mathbb{X}_{\bar{T}_h}^{\mathrm{S}_v}}{n}\left(\widehat{\ubeta}^{(t)}_{(\bar{T}_h)}-\bar{\ubeta}^{(t)}_{(\bar{T}_h)}\right) \\
			&\quad -\frac{1}{n}\left(\frac{1}{|\mathcal{A}_h|} \sum_{j \in \mathcal{A}_h} \nabla \mathcal{L}_j(\widehat{\ubeta}^{(t)}) - \nabla \mathcal{L}_{v}(\widehat{\ubeta}^{(t)})\right)\Bigg) \\
			&= \mathcal{B}\mathrm{sign}(\bar{\ubeta}^{(t)}_{(\bar{T}_h)})+\mathbb{X}_{\bar{T}_h^c}^{\mathrm{S}_v\top}(\mathbb{I}-\Pi_{X_T})\left(\frac{\epsilon^{\mathrm{S}_v}}{\lambda_{t}n}\right)+\frac{1}{n\lambda_{t}}\left(\frac{1}{|\mathcal{A}_h|} \sum_{j \in \mathcal{A}_h} \nabla \mathcal{L}_j(\widehat{\ubeta}^{(t)}) - \nabla \mathcal{L}_{v}(\widehat{\ubeta}^{(t)})\right).
		\end{split}
	\end{align}
	where $\mathcal{B}=\mathbb{X}_{\bar{T}_h^c}^{\mathrm{S}_v\top}\mathbb{X}_{\bar{T}_h}^{\mathrm{S}_v}\l(\mathbb{X}_{\bar{T}_h}^{\mathrm{S}_v\top}\mathbb{X}_{\bar{T}_h}^{\mathrm{S}_v}\r)^{-1}$, $\Pi_{X_T}=\mathbb{X}_{\bar{T}_h}\l(\mathbb{X}_{\bar{T}_h}^{\mathrm{S}_v\top}\mathbb{X}_{\bar{T}_h}^{\mathrm{S}_v}\r)^{-1}\mathbb{X}_{\bar{T}_h}^{\mathrm{S}_v\top}$. According to the proof of Theorem 11.3 in \cite{hastie2015statistical} and under Condition C1, we have:
	\begin{align}\label{lem5:eq3}
		\begin{split}
			\|\mathcal{B}\mathrm{sign}(\bar{\ubeta}^{(t)}_{\bar{T}_h})\|_{\infty} &< 1,\\
			\left\|\frac{1}{n\lambda_{t}}\mathbb{X}_{\bar{T}_h^c}^{\mathrm{S}_v^\top}(\mathbb{I}-\Pi_{X_T})\left(\frac{\epsilon^{\mathrm{S}_v}}{\lambda_{t} n}\right)\right\|_{\infty} &\xrightarrow{P} 0.
		\end{split}
	\end{align}
	By Lemma \ref{lem:dl}, it holds with probability approaching 1 as $n \to \infty$ that:
	\begin{align}\label{lem5:eq4}
		\begin{split}
			\left\| \frac{1}{|\mathcal{A}_h|} \sum_{j \in \mathcal{A}_h} \nabla \mathcal{L}_j(\widehat{\ubeta}^{(t)}) - \nabla \mathcal{L}_{v}(\widehat{\ubeta}^{(t)}) \right\|_{\infty} 
			\leq &\left\| \frac{1}{|\mathcal{A}_h|} \sum_{j \in \mathcal{A}_h} \nabla \mathcal{L}_j(\bar{\ubeta}^{\mathcal{A}_h}) \right\|_{\infty} \\
			&+ \frac{\log(np)}{n} \sqrt{\frac{\log(2p)}{n}} \| \bar{\ubeta}^{\mathcal{A}_h} - \widehat{\ubeta}^{(t)} \|_1 \\
			&+ C\left(\frac{\log(np)}{n}\right)^{2/3} \| \bar{\ubeta}^{\mathcal{A}_h} - \widehat{\ubeta}^{(t)} \|_1^2,
		\end{split}
	\end{align}
	for some universal constant $C$. By Lemma \ref{lem:dl},
	$$
	\| \widehat{\ubeta}^{(t)} - \bar{\ubeta}^{\mathcal{A}_h} \|_1 = O_P\left( \bar{s} \sqrt{\frac{\log p}{|\mathcal{A}_h|n}} + \left(\bar{s} \sqrt{\frac{\log p}{n}} \right)^{t+1} \right),
	$$
	thus, by Condition C3, the last two terms in (\ref{lem5:eq4}) are $o_P\left(\sqrt{\frac{\log p}{n}}\right)$.
	
	Notice that under condition C1-C2, it holds with probability to 1 that,
	$$
	\frac{\|\mathbb{X}^{\top}\epsilon\|_{\infty}}{n} \leq 2\sigma_{\epsilon}K_{clm}\sqrt{\frac{\log p}{n}},
	$$
	thus
	\begin{align}\label{lem5:eq6}
		\begin{split}
			&\left\| \frac{1}{|\mathcal{A}_h|} \sum_{j \in \mathcal{A}_h} \nabla \mathcal{L}_j(\bar{\ubeta}^{\mathcal{A}_h}) \right\|_{\infty} \leq 2\sigma_{\epsilon}K_{clm}\sqrt{\frac{\log p}{n}}.
		\end{split}
	\end{align}
	
	Therefore, we have with probability approaching 1 as $n \to \infty$:
	\begin{align}\label{lem5:eq5}
		\begin{split}
			\left\| \frac{1}{|\mathcal{A}_h|} \sum_{j \in \mathcal{A}_h} \nabla \mathcal{L}_j(\widehat{\ubeta}^{(t)}) - \nabla \mathcal{L}_{v}(\widehat{\ubeta}^{(t)}) \right\|_{\infty} 
			\leq 3\sigma_{\epsilon}K_{clm}\sqrt{\frac{\log p}{n}}.
		\end{split}
	\end{align}
	By (\ref{l5:eq2}), (\ref{lem5:eq3}) and (\ref{lem5:eq5}), under condition C3, we have that with probability approach 1 as $n\to\infty$,
	$$
	\| z^{(\Delta)}_{(T^c)}\|_{\infty}<1.
	$$
	By the proof of Theorem 11.3 in \cite{hastie2015statistical}, under condition C1-2, with probability at least $1-2\exp\{-c_2\lambda_{t}^2n\}$ for some constant $c_2$, 
	\begin{align}\label{l5:eq7}
		\begin{split}
			&\l\|\l(\frac{\mathbb{X}_{\bar{T}_h}^{\mathrm{S}_v\top} \mathbb{X}_{\bar{T}_h}^{\mathrm{S}_v}}{n}\r)^{-1}\frac{\mathbb{X}_{\bar{T}_h}^{\mathrm{S}_v\top}\epsilon^{\mathrm{S}_v} }{n}-\lambda_{t}\l(\frac{\mathbb{X}_{\bar{T}_h}^{\mathrm{S}_v\top} \mathbb{X}_{\bar{T}_h}^{\mathrm{S}_v}}{n}\r)^{-1}\mathrm{sign}(\bar{\ubeta}^{(t)}_{(\bar{T}_h)})\r\|_{\infty}\\
			& \leq 4\frac{1}{\sqrt{\bar{C}_{\min}}}\sigma_{\epsilon}\lambda_{t}+\l\|\l(\frac{\mathbb{X}_{\bar{T}_h}^{\mathrm{S}_v\top} \mathbb{X}_{\bar{T}_h}^{\mathrm{S}_v}}{n}\r)^{-1}\mathrm{sign}(\bar{\ubeta}^{(t)}_{(\bar{T}_h)})\r\|_{\infty}\lambda_{t}.
		\end{split}
	\end{align}
	Combine (\ref{l5:eq1}), (\ref{lem5:eq6}) , (\ref{l5:eq7}), we have that with probability approaching 1 as $n\to \infty$,
	\begin{align}\label{l5:eq8}
		\begin{split}
			\| \widehat{\ubeta}^{(t)} - \bar{\ubeta}^{\mathcal{A}_h} \|_{\infty}
			\leq & 3\sigma_{\epsilon}K_{clm}\sqrt{\frac{\log p}{n}}+4\frac{1}{\sqrt{\bar{C}_{\min}}}\sigma_{\epsilon}\lambda_{t}\\
			&+\quad \l\|\l(\frac{\mathbb{X}_{\bar{T}_h}^{\mathrm{S}_v\top} \mathbb{X}_{\bar{T}_h}^{\mathrm{S}_v}}{n}\r)^{-1}\mathrm{sign}(\bar{\ubeta}^{(t)}_{(\bar{T}_h)})\r\|_{\infty}\lambda_{t}.
		\end{split}
	\end{align}
\end{proof}
	
	{\bf Proof of Lemma \ref{Lemma:3}:}
According to the zero-subgradient conditions, we have:
\begin{align}\label{kkt}
	\begin{split}
		-\frac{1}{n}\mathbb{X}^\top (\mathbb{Y} -\mathbb{X} \hat{\ubeta}^{\mathcal{A}_h} - \mathbb{X} \hat{\Delta}^{\mathcal{A}_h} - \sqrt{n}\hat{e}^{\mathcal{A}_h}) + \lambda_{\Delta} z^{(\Delta)} &= 0,\\
		-\frac{1}{\sqrt{n}} (\mathbb{Y} -\mathbb{X} \hat{\ubeta}^{\mathcal{A}_h} - \mathbb{X} \hat{\Delta}^{\mathcal{A}_h} - \sqrt{n}\hat{e}^{\mathcal{A}_h}) + \lambda_{e} z^{(e)} &= 0,
	\end{split}
\end{align}
where $z^{(e)} \in \partial \|e\|_1$ and $z^{(\Delta)} \in \partial \|\Delta\|_1$ are subgradients. Denote $T$ and $E$ be the support sets of $\Delta^{\mathcal{A}_h}$ and $e^*$, respectively. We apply the primal-dual witness method(\cite{mj2009sharp}) as follows:

(i) Set $\hat{\Delta}^{\mathcal{A}_h}_{(T^c)}=0$, $\hat{e}^{\mathcal{A}_h}_{(E^c)}=0$.

(ii) Determine $\hat{e}^{\mathcal{A}_h}_E, z^{(e)}, \hat{\Delta}^{\mathcal{A}_h}_{(T)}, z^{(\Delta)}$ by solving (\ref{kkt}).

(iii) Check whether the strict dual feasibility conditions $\|z^{(e)}_{(E^c)}\|_{\infty}<1$ and $\|z^{(\Delta)}_{(T^c)}\|_{\infty}<1$ hold.

Rewriting the zero-subgradient conditions (\ref{kkt}) in block matrix form yields:
$$
\frac{1}{n}
\mathcal{A} 
\begin{bmatrix}
	\hat{\Delta}^{\mathcal{A}_h}_{(T)}-\Delta^{\mathcal{A}_h}_{(T)} \\
	0
\end{bmatrix}
- \frac{1}{n}
\begin{bmatrix}
	\mathbb{X}_{T}^\top (\epsilon +\mathcal{T})+(\widehat{e}^{\mathcal{A}_h} - e^*)\\
	\mathbb{X}_{T^c}^\top (\epsilon+\mathcal{T})+(\widehat{e}^{\mathcal{A}_h} - e^*)
\end{bmatrix}
+ \lambda_{\Delta}
\begin{bmatrix}
	z^{(\Delta)}_{T} \\
	z^{(\Delta)}_{(T^c)}
\end{bmatrix}
=
\begin{bmatrix}
	0 \\
	0
\end{bmatrix}.
$$

Similarly for the $e$ terms:
$$
\begin{bmatrix}
	\widehat{e}_{E}^{\mathcal{A}_h} - e_{(E)}^* \\
	0
\end{bmatrix}
-
\frac{1}{\sqrt{n}}
\begin{bmatrix}
	\mathcal{T}_{E}+\mathbb{X}_{ET}(\hat{\Delta}^{\mathcal{A}_h}_{(T)}-\Delta^{\mathcal{A}_h}_{(T)})+\epsilon_{E}   \\
	\mathcal{T}_{(E^c)}+\mathbb{X}_{E^cT}(\hat{\Delta}^{\mathcal{A}_h}_{(T)}-\Delta^{\mathcal{A}_h}_{(T)})+\epsilon_{(E^c)}
\end{bmatrix}
+
\lambda_{e}
\begin{bmatrix}
	z^{(e)}_{E} \\
	z^{(e)}_{(E^c)}
\end{bmatrix}
=
\begin{bmatrix}
	0 \\
	0
\end{bmatrix},
$$
where
$$
\mathcal{A} = \begin{bmatrix}
	\mathbb{X}_{T}^\top \mathbb{X}_{T} & \mathbb{X}_{T}^\top \mathbb{X}_{T^c} \\
	\mathbb{X}_{T}^\top \mathbb{X}_{T^c} & \mathbb{X}_{T^c}^\top \mathbb{X}_{T^c}
\end{bmatrix}, \quad \mathcal{T} = \mathbb{X}(\hat{\ubeta}^{\mathcal{A}_h}-\bar{\ubeta}^{\mathcal{A}_h}).
$$

By \cite{Tibshirani2012TheLP}, under absolutely continuous distribution condition(C1), the solution for $\hat{\Delta}^{\mathcal{A}_h}$ is unique. Under absolutely continuous distribution condition(C1) and the condition that $\Sigma$ has minimum eigenvalue bounded from 0(C1), $\mathbb{X}_{T}^\top \mathbb{X}_{T}$ is invertible, satisfying
$$
C_{\min} := \lambda_{\min}\left(\frac{\mathbb{X}_{T}^\top \mathbb{X}_{T}}{n}\right) > 0.
$$

Solving for the vectors $\hat{\Delta}^{\mathcal{A}_h}_{(T)}-\Delta^{\mathcal{A}_h}_{(T)}$ and $\hat{e}^{\mathcal{A}_h}_{(E)}-e_{(E)}^*$ result in:
\begin{align}\label{l3:eq1}
	\begin{split}
		\hat{\Delta}^{\mathcal{A}_h}_{(T)}-\Delta^{\mathcal{A}_h}_{(T)}=&\left(\frac{\mathbb{X}_{T}^\top \mathbb{X}_{T}}{n}\right)^{-1}\frac{\mathbb{X}_{T}^\top (\epsilon +\mathcal{T})}{n}-\lambda_{\Delta}\left(\frac{\mathbb{X}_{T}^\top \mathbb{X}_{T}}{n}\right)^{-1}\mathrm{sign}(\Delta^{\mathcal{A}_h}_{(T)})\\
		&+\left(\frac{\mathbb{X}_{T}^\top \mathbb{X}_{T}}{n}\right)^{-1}\frac{1}{n}\mathbb{X}_{T}^\top(\widehat{e}^{\mathcal{A}_h} - e^*),\\
		\hat{e}^{\mathcal{A}_h}_{(E)}-e_{(E)}^*=&\frac{1}{\sqrt{n}}\mathcal{T}_{E}+\frac{1}{\sqrt{n}}\mathbb{X}_{ET}(\hat{\Delta}^{\mathcal{A}_h}_{(T)}-\Delta^{\mathcal{A}_h}_{(T)})+\frac{1}{\sqrt{n}} \epsilon_{E}-\lambda_{e}\mathrm{sign}(e^*_{(E)}).
	\end{split}
\end{align}
Solving for the vectors $z^{(\Delta)}_{(T^c)}$ and $z^{(e)}_{(E^c)}$ yields:
\begin{align}\label{l3:eq2}
	\begin{split}
		z^{(\Delta)}_{(T^c)} &= \frac{1}{\lambda_{\Delta}}\left[\frac{1}{n}\mathbb{X}_{T^c}^\top (\epsilon + \mathcal{T} + \widehat{e}^{\mathcal{A}_h} - e^*) - \frac{\mathbb{X}_{T^c}^\top \mathbb{X}_{T}}{n}(\hat{\Delta}^{\mathcal{A}_h}_{(T)} - \Delta^{\mathcal{A}_h}_{(T)})\right]\\
		&= \mathcal{B}\mathrm{sign}(\Delta^{\mathcal{A}_h}_{(T)}) + \mathbb{X}_{T^c}^\top(\mathbb{I} - \Pi_{X_T})\left(\frac{\epsilon + \mathcal{T} + \widehat{e}^{\mathcal{A}_h} - e^*}{\lambda_{\Delta} n}\right),\\
		z^{(e)}_{(E^c)} &= \frac{1}{\lambda_e}\left[\frac{1}{\sqrt{n}}\mathcal{T}_{(E^c)} + \frac{1}{\sqrt{n}}\mathbb{X}_{E^cT}(\hat{\Delta}^{\mathcal{A}_h}_{(T)} - \Delta^{\mathcal{A}_h}_{(T)}) + \frac{1}{\sqrt{n}}\epsilon_{(E^c)}\right],
	\end{split}
\end{align}
where $\mathbb{I}$ is the identity matrix, $\mathcal{B} = \mathbb{X}_{T^c}^\top \mathbb{X}_{T}(\mathbb{X}_{T}^\top \mathbb{X}_{T})^{-1}$ and $\Pi_{X_T} = \mathbb{X}_{T}(\mathbb{X}_{T}^\top \mathbb{X}_{T})^{-1}\mathbb{X}_{T}^\top$. According to the proof of Theorem 11.3 in \cite{hastie2015statistical}, under condition C1-C2, we have:
\begin{align}\label{l3:eq2_1}
	\begin{split}
		\|\mathcal{B}\mathrm{sign}(\Delta^{\mathcal{A}_h}_{(T)})\|_{\infty} &< 1,\\
		\left\|\frac{1}{n\lambda_{\Delta}}\mathbb{X}_{T^c}^\top(\mathbb{I} - \Pi_{X_T})\left(\frac{\epsilon}{\lambda_{\Delta} n}\right)\right\|_{\infty} &\xrightarrow{P} 0.
	\end{split}
\end{align}
By (\ref{rec0}) and Lemma \ref{Lemma:2},
\begin{align}\label{e}
	\begin{split}
	\|\widehat{e}^{\mathcal{A}_h} - e^*\|_1&\leq
4\max\left(\sqrt{s_{\Delta}}\lambda_{\Delta}/\lambda_e,\sqrt{k}\right)\left(\sqrt{\frac{s_{\Delta}\log p}{n}}+\sqrt{\frac{k\log n}{n}}\right) + C\frac{\bar{s}\sqrt{\log p}}{|\mathcal{A}_h|\sqrt{n}}\\
&=O_P\l(\frac{s_{\Delta}\log p}{\sqrt{n\log n}}+k\sqrt{\frac{\log n}{n}} +\frac{\bar{s}\sqrt{\log p}}{|\mathcal{A}_h|\sqrt{n}}\r).
\end{split}
\end{align}

It holds that for $j\in\mathcal{A}_h$,
$$
h\geq \| \Delta^{(S_j)} \|_1 \geq \l(\min_{j : \ubeta_j^* \neq 0} |\ubeta_j^*|\wedge  \min_{j : \ubeta_j^{(\mathrm{S}_j)} \neq 0} |\ubeta_j^{(\mathrm{S}_j)}|\r)\| \Delta^{(S_j)} \|_0,
$$
thus  by condition C3-4,
$$
s_{\Delta}\leq \sum_{j\in\mathcal{A}_h}  \|\Delta^{(S_j)} \|_0\leq (\gamma_1\wedge\gamma_2)^{-1}h|\mathcal{A}_h |,
$$
and $\bar{s} \leq \gamma_2^{-1}|\mathcal{A}_h|h$.
By Lemma \ref{Lemma:1}, with probability to 1,
$$
\left\|\mathcal{T}\right\|_1 \leq \sqrt{n}\left\|\mathcal{T}\right\|_2 \leq\sqrt{\frac{49 \bar{s} }{2} \frac{\log(p)}{|\mathcal{A}_h|}}.
$$
Since it holds with probability to 1 that,
$$
\l\|\mathbb{X}_{T^c}^\top(\mathbb{I} - \Pi_{X_T})\r\|_{\infty}\leq \l\|\mathbb{X}_{T^c}^\top\r\|_{\infty}\leq 2\|\mathrm{diag}(\Sigma)\|_{\infty}\sqrt{\log((p-s_{\delta})n)},
$$
by (\ref{e}), under condition C3, it holds with probability approaching 1 as $n \to \infty$ that, 
\begin{align}\label{l3:eq2_2}
	\begin{split}
		&\l\|\mathbb{X}_{T^c}^\top(\mathbb{I} - \Pi_{X_T})\frac{\mathcal{T} + \widehat{e}^{\mathcal{A}_h} - e^*}{\lambda_{\Delta} n}\r\|_{\infty}\\
		&\leq \frac{\|\mathbb{X}_{T^c}\|_{\infty}}{\lambda_{\Delta} n}\left[  \|\mathcal{T}\|_1 + \|\widehat{e}^{\mathcal{A}_h} - e^*\|_1\right]\\
		&=o_P(1).
	\end{split}
\end{align}

Combining (\ref{l3:eq2}), (\ref{l3:eq2_1}), and (\ref{l3:eq2_2}), it follows that with probability approaching 1 as $n \to \infty$,
$$
\|z^{(\Delta)}_{(T^c)}\|_{\infty} < 1.
$$

With probability to 1, by Lemma \ref{Lemma:5} and the mutual incoherence condition(C1), it holds that,
\begin{align}\label{l3:eq3}
	\begin{split}
		&\l\|\l(\frac{\mathbb{X}_{T}^{\top} \mathbb{X}_{T}}{n}\r)^{-1}\frac{\mathbb{X}_{T}^{\top}  \mathcal{T}}{n}\r\|_{\infty}\\
		&\leq
		\|\hat{\ubeta}_T^{\mathcal{A}_h}-\bar{\ubeta}_T^{\mathcal{A}_h}\|_{\infty} + \l\|\l(\frac{\mathbb{X}_{T}^{\top} \mathbb{X}_{T}}{n}\r)^{-1}\frac{\mathbb{X}_{T}^{\top}  \mathbb{X}_{T^c}\l(\hat{\ubeta}_{T^c}^{\mathcal{A}_h}-\bar{\ubeta}_{T^c}^{\mathcal{A}_h} \r)}{n}\r\|_{\infty}
		\\
		& \leq 	(2-\gamma)\left\|\hat{\ubeta}^{\mathcal{A}_h} - \bar{\ubeta}^{\mathcal{A}_h}\right\|_{\infty} \\
		&\leq 6\sigma_{\epsilon}K_{clm}\sqrt{\frac{\log p}{n}}+8\frac{1}{\sqrt{\bar{C}_{\min}}}\sigma_{\epsilon}\lambda_{t}
		+ 2\l\|\l(\frac{\mathbb{X}_{\bar{T}_h}^{\mathrm{S}_v\top} \mathbb{X}_{\bar{T}_h}^{\mathrm{S}_v}}{n}\r)^{-1}\mathrm{sign}(\bar{\ubeta}^{(t)}_{(\bar{T}_h)})\r\|_{\infty}\lambda_{t}.
	\end{split}
\end{align}

According to the proof of Theorem 11.3 in \cite{hastie2015statistical}, with probability at least $1 - 2\exp\{-c_2 \lambda_{\Delta}^2 n\}$ for some constant $c_2$, we have:
\begin{align}\label{l3:eq4}
	\begin{split}
		&\left\|\left(\frac{\mathbb{X}_{T}^\top \mathbb{X}_{T}}{n}\right)^{-1}\frac{\mathbb{X}_{T}^\top \epsilon}{n} - \lambda_{\Delta}\left(\frac{\mathbb{X}_{T}^\top \mathbb{X}_{T}}{n}\right)^{-1}\mathrm{sign}(\Delta^{\mathcal{A}_h}_{(T)})\right\|_{\infty} \\
		&\leq 4\frac{1}{\sqrt{C_{\min}}} \sigma_{\epsilon} \lambda_{\Delta} + \left\|\left(\frac{\mathbb{X}_{T}^\top \mathbb{X}_{T}}{n}\right)^{-1}\mathrm{sign}(\Delta^{\mathcal{A}_h}_{(T)})\right\|_{\infty} \lambda_{\Delta}.
	\end{split}
\end{align}
By \ref{rec0}, with probability approach 1 as $n\to \infty$, under condition C1, it holds that, 
\begin{align}\label{l3:eq4_1}
	\begin{split}
		&\left\|\left(\frac{\mathbb{X}_{T}^{\top} \mathbb{X}_{T}}{n}\right)^{-1}\frac{1}{n}\mathbb{X}_{T}^{\top}\left(\widehat{e}^{\mathcal{A}_h} - e^* \right)\right\|_{\infty} \\
		&\quad \leq \left\|\left(\frac{\mathbb{X}_{T}^{\top} \mathbb{X}_{T}}{n}\right)^{-1}\right\|_{L_1}\left\|\left(\frac{1}{n}\mathbb{X}_{T}^{\top}\right)\right\|_{\infty}\|\widehat{e}^{\mathcal{A}_h} - e^* \|_1\\
		&\quad \leq \left\|\left(\frac{\mathbb{X}_{T}^{\top} \mathbb{X}_{T}}{n}\right)^{-1}\right\|_{L_1}\frac{\|\Sigma\|_{\infty}\sqrt{\log(s_{\Delta}n)}}{n} \left(4\max\left(\sqrt{s_{\Delta}}\lambda_{\Delta}/\lambda_e,\sqrt{k}\right)\right.\\
		&\qquad\left.\cdot\left(\sqrt{\frac{s_{\Delta}\log p}{n}}+\sqrt{\frac{k\log n}{n}}\right) + C\frac{\bar{s}\log p}{|\mathcal{A}_h|n}\right),
	\end{split}
\end{align}
Combine (\ref{l3:eq1}), (\ref{l3:eq3}) , (\ref{l3:eq4}), and (\ref{l3:eq4_1}), we have that with probability approaching 1 as $n\to \infty$,
\begin{align}\label{l3:eq5}
	\begin{split}
		&\|\hat{\Delta}^{\mathcal{A}_h}-\Delta^{\mathcal{A}_h}\|_{\infty} \\
		&\quad \leq 6\sigma_{\epsilon}K_{clm}\sqrt{\frac{\log p}{n}}+8\frac{1}{\sqrt{\bar{C}_{\min}}}\sigma_{\epsilon}\lambda_{t}\\
		&\qquad+ 2\l\|\l(\frac{\mathbb{X}_{\bar{T}_h}^{\mathrm{S}_v\top} \mathbb{X}_{\bar{T}_h}^{\mathrm{S}_v}}{n}\r)^{-1}\mathrm{sign}(\bar{\ubeta}^{(t)}_{(\bar{T}_h)})\r\|_{\infty}\lambda_{t} \\
		&\quad\quad + 4\frac{1}{\sqrt{C_{\min}}}\sigma_{\epsilon}\lambda_{\Delta} + \left\|\left(\frac{\mathbb{X}_{T}^\top \mathbb{X}_{T}}{n}\right)^{-1}\mathrm{sign}(\Delta^{\mathcal{A}_h}_{(T)})\right\|_{\infty}\lambda_{\Delta} \\
		&\quad\quad + \left\|\left(\frac{\mathbb{X}_{T}^{\top} \mathbb{X}_{T}}{n}\right)^{-1}\right\|_{L_1}\frac{\|\Sigma\|_{\infty}\sqrt{\log(s_{\Delta}n)}}{n} \times \\
		&\qquad\quad \left[4\max\left(\sqrt{s_{\Delta}}\lambda_{\Delta}/\lambda_e,\sqrt{k}\right) \left(\sqrt{\frac{s_{\Delta}\log p}{n}}+\sqrt{\frac{k\log n}{n}}\right)\right. \\
		&\qquad \quad\left.+ C\frac{\bar{s}\log p}{|\mathcal{A}_h|n}\right],
	\end{split}
\end{align}
By combining Lemma \ref{Lemma:5} with equation (\ref{l3:eq5}), we complete the proof of the first part of this lemma.

Observe that 
\begin{align*}
	&4\max\left(\sqrt{s_{\Delta}}\lambda_{\Delta}/\lambda_e,\sqrt{k}\right) \left(\sqrt{\frac{s_{\Delta}\log p}{n}}+\sqrt{\frac{k\log n}{n}}\right)\\
	&\quad =O_P\l(\frac{s_{\Delta}\log p}{\sqrt{n\log n}}+k\sqrt{\frac{\log n}{n}} \r).
\end{align*}
By Lemma 5 of \cite{mj2009sharp}, it holds that,
\begin{align*}
	\left\|\left(\frac{\mathbb{X}_{T}^{\top} \mathbb{X}_{T}}{n}\right)^{-1}\right\|_{L_1} &\leq \left\|\sqrt{\Omega_{TT}}\right\|_{L_1}^2, \\
	\left\|\left(\frac{\mathbb{X}_{\bar{T}_h}^{\mathrm{S}_v\top} \mathbb{X}_{\bar{T}_h}^{\mathrm{S}_v}}{n}\right)^{-1}\mathrm{sign}(\bar{\beta}^{\mathcal{A}_h}_{(\bar{T}_h)})\right\|_{\infty} &\leq \left\|\sqrt{\Omega_{TT}^{\mathrm{S}_v}}\right\|_{L_1}^2, \\
	\left\|\left(\frac{\mathbb{X}_{T}^{\top} \mathbb{X}_{T}}{n}\right)^{-1}\mathrm{sign}(\Delta^{\mathcal{A}_h}_{(T)})\right\|_{\infty} &\leq \left\|\sqrt{\Omega_{TT}}\right\|_{L_1}^2.
\end{align*}

Put these pieces together, by condition C3, 
we have that, with probability approaching 1 as $n\to \infty$,
\begin{align}\label{l3:eq51}
	\begin{split}
		&\|\hat{\Delta}^{\mathcal{A}_h}-\Delta^{\mathcal{A}_h}\|_{\infty} \\
		&\quad \leq 6\sigma_{\epsilon}K_{clm}\sqrt{\frac{\log p}{n}}+8\frac{1}{\sqrt{\bar{C}_{\min}}}\sigma_{\epsilon}\lambda_{t} \\
		&\quad\quad + 2\left\|\left(\frac{\mathbb{X}_{\bar{T}_h}^{\mathrm{S}_v\top} \mathbb{X}_{\bar{T}_h}^{\mathrm{S}_v}}{n}\right)^{-1}\mathrm{sign}(\bar{\ubeta}^{(t)}_{(\bar{T}_h)})\right\|_{\infty}\lambda_{t} \\
		&\quad\quad + 4\frac{1}{\sqrt{C_{\min}}}\sigma_{\epsilon}\lambda_{\Delta} + \left\|\left(\frac{\mathbb{X}_{T}^\top \mathbb{X}_{T}}{n}\right)^{-1}\mathrm{sign}(\Delta^{\mathcal{A}_h}_{(T)})\right\|_{\infty}\lambda_{\Delta} \\
		&\quad\quad +o_P\l(\sqrt{\frac{\log p}{n}}\r)\\
		&\quad=O_P\l(\sqrt{\frac{\log p}{n}}\r).
	\end{split}
\end{align}
By (\ref{l3:eq2}), we have that,
\begin{align}\label{rlasso3}
	\begin{split}
		\lambda_e\|z^{(e)}_{(E^c)}\|_{\infty} \leq& \l\|\frac{1}{\sqrt{n}}\mathbb{X}_{r(E)}\r\|_{\infty}\l\|\hat{\ubeta}^{\mathcal{A}_h}-\bar{\ubeta}^{\mathcal{A}_h})\r\|_{1}+ \l\|\frac{1}{\sqrt{n}}\mathbb{X}_{r(E)}(\hat{\Delta}^{\mathcal{A}_h}-\Delta^{\mathcal{A}_h})\r\|_{\infty}+\l\|\frac{1}{\sqrt{n}} \epsilon_{(E)}\r\|_{\infty}.
	\end{split}
\end{align}
By condition C1,
$$
\frac{1}{\sqrt{n}}\|\mathbb{X}_{r(E^c)}\|_{\infty}=O_P\l(\sqrt{\frac{\log((n-k) p)}{n}}\r).
$$
Thus by Lemma \ref{lem:dl}, for large enough $t$, it holds that
\begin{align}\label{eq19}
	\begin{split}
		\|\frac{1}{\sqrt{n}}\mathbb{X}_{r(E^c)}\|_{\infty}\|\hat{\ubeta}^{\mathcal{A}_h}-\bar{\ubeta}^{\mathcal{A}_h})\|_{1}=O_P\l(\bar{s} \frac{\sqrt{\log((n-k)p)\log p}}{\sqrt{|\mathcal{A}_h|}n}\r).
	\end{split}
\end{align}
By (\ref{l3:eq51}),
\begin{align}\label{eq20}
	\begin{split}
		\l\|\frac{1}{\sqrt{n}}\mathbb{X}_{E^cT}(\hat{\Delta}^{\mathcal{A}_h}_{(T)}-\Delta^{\mathcal{A}_h}_{(T)})\r\|_{\infty}&\leq \l\|\frac{1}{\sqrt{n}}\mathbb{X}_{E^cT}\r\|_{\infty}s_{\Delta}\l\|\hat{\Delta}^{\mathcal{A}_h}_{(T)}-\Delta^{\mathcal{A}_h}_{(T)}\r\|_{\infty}\\
		&= O_P\l(s_{\Delta}\frac{\sqrt{\log((n-k)s_{\Delta})\log p}}{n}\r).
	\end{split}
\end{align}
By (\ref{rlasso3}), (\ref{eq19}) and (\ref{eq20}), and the choice of $\lambda_e$, under condition C3, it holds that with probability  to 1,
$$
\|z^{(e)}_{(E^c)}\|_{\infty}<1.
$$
By a similar arguments, under condition C3, it holds that,
\begin{align*}
	\|\hat{e}^{\mathcal{A}_h}-e^*\|_{\infty} \leq& \|\frac{1}{\sqrt{n}}\|\mathbb{X}_{r(E)}\|_{\infty}\| \hat{\ubeta}^{\mathcal{A}_h} -\bar{\ubeta}^{\mathcal{A}_h}\|_1+\l\|\frac{1}{\sqrt{n}}\mathbb{X}_{ET}\r\|_{\infty}s_{\Delta}\l\|\hat{\Delta}^{\mathcal{A}_h}_{(T)}-\Delta^{\mathcal{A}_h}_{(T)}\r\|_{\infty}\\
	&+ \|\frac{1}{\sqrt{n}} \epsilon_{(E)}\|_{\infty}+\lambda_{e}\\
	=&O_P\l(\bar{s} \frac{\sqrt{\log(kp)\log p}}{\sqrt{|\mathcal{A}_h|}n}+ s_{\Delta}\frac{\sqrt{\log(kp)\log p}}{n}+3\sigma_{\epsilon}\sqrt{\frac{\log n}{n}}\r)\\
	=& O_P\l(\sqrt{\frac{\log n}{n}}\r).
\end{align*}
The proof is completed.

	{\bf Proof of Proposition \ref{prop0}:}
		Define $z=\hat{\Delta}^{\mathcal{A}_h}-\Delta^{\mathcal{A}_h}$ and $v=\hat{e}^{\mathcal{A}_h}-e^*$.
	For the proof of the first part, observe that,
	\begin{equation}\label{eq:m00}
		\begin{aligned}
			\frac{1}{2n} \|\mathbb{X}z + \sqrt{n} v\|_2^2 
			&\leq \frac{1}{n} \langle y - \mathbb{X} \hat{\ubeta}^{\mathcal{A}_h} - \mathbb{X} \Delta^{\mathcal{A}} - \sqrt{n} e^*, \mathbb{X}z + \sqrt{n} v \rangle \\
			&\quad+  \lambda_{\Delta} (\|\Delta^{\mathcal{A}_h}\|_1 - \|\hat{\Delta}^{\mathcal{A}}\|_1) + \lambda_{e} (\|e^*\|_1 - \|\hat{e}^{\mathcal{A}_h}\|_1)\\
			&= \frac{1}{n} \langle \epsilon, \mathbb{X}z + \sqrt{n} v \rangle + \lambda_{\Delta} (\|\Delta^{\mathcal{A}_h}\|_1 - \|\hat{\Delta}^{\mathcal{A}}\|_1) \\
			&\quad+  \frac{1}{n}\langle \mathbb{X}(\bar{\ubeta}^{\mathcal{A}_h}-\hat{\ubeta}^{\mathcal{A}_h}), \mathbb{X}z + \sqrt{n} v \rangle + \lambda_{e} (\|e^*\|_1 - \|\hat{e}^{\mathcal{A}_h}\|_1) \\
			&\leq \frac{1}{n} \|\mathbb{X}^\top \epsilon\|_\infty \|z\|_1 + \frac{1}{\sqrt{n}} \|\epsilon\|_\infty \|v\|_1 \\
			&\quad+ \lambda_{\Delta} (\|\Delta^{\mathcal{A}_h}\|_1 - \|\hat{\Delta}^{\mathcal{A}}\|_1) + \lambda_{e} (\|e^*\|_1 - \|\hat{e}^{\mathcal{A}_h}\|_1) \\
			&\quad+ \frac{1}{4n} \|\mathbb{X}z + \sqrt{n} v\|_2^2+ \frac{1}{n} \|\mathbb{X}(\bar{\ubeta}^{\mathcal{A}_h}-\hat{\ubeta}^{\mathcal{A}_h})\|_2^2 \\
			&\leq -\frac{1}{2}\lambda_{\Delta} \|z\|_1  + 2\lambda_{\Delta} \|\Delta^{\mathcal{A}_h}\|_1   \\    
			&\quad+ \left(\frac{1}{\sqrt{n}} \|\epsilon\|_\infty + \lambda_{e}\right) \|v_E\|_1- \left(\lambda_{e} - \frac{1}{\sqrt{n}} \|\epsilon\|_\infty \right) \|v_{E^c}\|_1\\
			&\quad+ \frac{1}{4n} \|\mathbb{X}z + \sqrt{n} v\|_2^2+ \frac{1}{n} \|\mathbb{X}(\bar{\ubeta}^{\mathcal{A}_h}-\hat{\ubeta}^{\mathcal{A}_h})\|_2^2 \\
		\end{aligned}
	\end{equation}
	By Lemma \ref{Lemma:3}, with probability to 1, it holds that
	\begin{equation}\label{eq:m0}
		\begin{aligned}
			\frac{1}{2n} \|\mathbb{X}z + \sqrt{n} v\|_2^2 
			&\leq -\frac{1}{2}\lambda_{\Delta} \|z\|_1  + 2\lambda_{\Delta} \|\Delta^{\mathcal{A}_h}\|_1   \\    
			&\quad+C \left(\frac{1}{\sqrt{n}} \|\epsilon\|_\infty + \lambda_{e}\right)k \sqrt{\frac{\log n}{n}}\\
			&\quad+ \frac{1}{4n} \|\mathbb{X}z + \sqrt{n} v\|_2^2+ \frac{1}{n} \|\mathbb{X}(\bar{\ubeta}^{\mathcal{A}_h}-\hat{\ubeta}^{\mathcal{A}_h})\|_2^2, \\
		\end{aligned}
	\end{equation}
	for some universal constant $C$.
	
	(i) if $\frac{1}{n} \|\mathbb{X}(\bar{\ubeta}^{\mathcal{A}_h}-\hat{\ubeta}^{\mathcal{A}_h})\|_2^2\leq \lambda_{\Delta} \|\Delta^{\mathcal{A}_h}\|_1$, by (\ref{eq:m0}), it holds with probability approaching 1 as $n\to \infty$,
	\begin{equation}\label{eq2:m1}
		\begin{aligned}
			\frac{1}{4n} \|\mathbb{X}z + \sqrt{n} v\|_2^2 
			&\leq -\frac{1}{2}\lambda_{\Delta} \|z\|_1  + 3\lambda_{\Delta} \|\Delta^{\mathcal{A}_h}\|_1 +C\frac{k\log n}{n}\\
			&\leq 3 C\l(\lambda_{\Delta}h\wedge h^2\r)+C\frac{k\log n}{n}  \\
		\end{aligned}
	\end{equation}
	for some universal constant $C$. By the extended RE condition,
	\begin{equation}\label{eq2:m2}
		\begin{aligned}
			\kappa_l^2(\|z\|_2 + \|v\|_2)^2 &\leq  3 C\l(\lambda_{\Delta}h\wedge h^2\r)  +C\frac{k\log n}{n}
		\end{aligned}
	\end{equation}
	(ii) if $\frac{1}{n} \|\mathbb{X}(\bar{\ubeta}^{\mathcal{A}_h}-\hat{\ubeta}^{\mathcal{A}_h})\|_2^2\geq \lambda_{\Delta} \|\Delta^{\mathcal{A}_h}\|_1$, by Lemma \ref{Lemma:1} and (\ref{eq:m0}),
	\begin{equation}\label{eq2:m3}
		\begin{aligned}
			\frac{1}{4n} \|\mathbb{X}z + \sqrt{n} v\|_2^2 
			&\leq -\frac{1}{2}\lambda_{\Delta} \|z\|_1  -\frac{1}{2}\lambda_{e} \|v\|_1 + \frac{2}{n} \|\mathbb{X}(\bar{\ubeta}^{\mathcal{A}_h}-\hat{\ubeta}^{\mathcal{A}_h})\|_2^2+C\frac{k\log n}{n}  \\
			&= O_P\l(\frac{\bar{s}\log p}{|\mathcal{A}_h|n}+\frac{k\log n}{n}\r)
		\end{aligned}
	\end{equation}
	by the extended RE condition,
	\begin{equation}\label{eq2:m4}
		\begin{aligned}
			\kappa_l^2(\|z\|_2 + \|v\|_2)^2 = O_P\l( \frac{\bar{s}\log p}{|\mathcal{A}_h|n}+\frac{k\log n}{n}\r)
		\end{aligned}
	\end{equation}
	combine (\ref{eq2:m2}) and (\ref{eq2:m4}),
	\begin{equation}\label{eq2:m5}
		\begin{aligned}
			\|z\|_2 + \|v\|_2 = O_P\l( \sqrt{\frac{\bar{s}\log p}{|\mathcal{A}_h|n}}+\sqrt{\lambda_{\Delta} h}\wedge h   +\sqrt{\frac{k\log n}{n}}\r)
		\end{aligned}
	\end{equation}
	By Lemma \ref{Lemma:2} and (\ref{eq2:m5}),
	\begin{equation}\label{eq2:m6}
		\begin{aligned}
			\|z\|_2 + \|v\|_2 = O_P\l( \sqrt{\frac{\bar{s}\log p}{|\mathcal{A}_h|n}}+\sqrt{\lambda_{\Delta} h}\wedge h \wedge\sqrt{\frac{s_{\Delta}\log p}{n}}+\sqrt{\frac{k\log n}{n}}\r)
		\end{aligned}
	\end{equation}
	By equation (\ref{eq2:m6}), we complete the proof of the first part.
	
	The proof of the second part is a direct consequence of Lemma \ref{Lemma:3} and Condition C4.
	
   Denote 
	\begin{equation}
		\begin{split}
		 r_n := & 9K_{clm}\hat{\sigma}_{\epsilon}\sqrt{\frac{\log p}{n}} + 12\frac{1}{\sqrt{\bar{C}_{\min}}}\hat{\sigma}_{\epsilon}\lambda_{t} \\
			& \quad + 3\left\|\left(\frac{\mathbb{X}_{\bar{T}_h}^{\mathrm{S}_v\top} \mathbb{X}_{\bar{T}_h}^{\mathrm{S}_v}}{n}\right)^{-1}\operatorname{sign}(\bar{\ubeta}^{\mathcal{A}_h}_{(\bar{T}_h)})\right\|_{\infty}\lambda_{t} \\
			& \quad + 4\frac{1}{\sqrt{C_{\min}}}\hat{\sigma}_{\epsilon}\lambda_{\Delta} + \left\|\left(\frac{\mathbb{X}_{T}^\top \mathbb{X}_{T}}{n}\right)^{-1}\operatorname{sign}(\Delta^{\mathcal{A}_h}_{(T)})\right\|_{\infty}\lambda_{\Delta},
		\end{split}
	\end{equation}
		By Lemma \ref{Lemma:3}, it holds with probability to 1 that
		$$
			\|\hat{\ubeta}^{\mathrm{oracle}}-\ubeta^*\|_{\infty}\leq(1+o(1))r_n.
		$$
	If the covariates follow a standard Gaussian distribution, then by the Marchenko–Pastur law (see, e.g., \cite{couillet2011random}), $\bar{C}_{\min}$ and $C_{\min}$ converge in distribution to 1. 	By Lemma 5 of \cite{mj2009sharp}, it holds that,
	\begin{align*}
		\left\|\left(\frac{\mathbb{X}_{\bar{T}_h}^{\mathrm{S}_v\top} \mathbb{X}_{\bar{T}_h}^{\mathrm{S}_v}}{n}\right)^{-1}\mathrm{sign}(\bar{\beta}^{\mathcal{A}_h}_{(\bar{T}_h)})\right\|_{\infty} &\leq \left\|\sqrt{\Omega_{TT}^{\mathrm{S}_v}}\right\|_{L_1}^2, \\
		\left\|\left(\frac{\mathbb{X}_{T}^{\top} \mathbb{X}_{T}}{n}\right)^{-1}\mathrm{sign}(\Delta^{\mathcal{A}_h}_{(T)})\right\|_{\infty} &\leq \left\|\sqrt{\Omega_{TT}}\right\|_{L_1}^2.
	\end{align*}
			By Theorem 9.3 in \cite{fan2020statistical},
	\begin{align*}
		\|\hat{\Sigma}^{\mathrm{S}_v}-\Sigma\|_{\infty} = O_P\left(\sqrt{\frac{\log p}{n}}\right).
	\end{align*}
	Furthermore, under standard Gaussian design, the noise level could be estimated consistently using well developed technique(see, e.g. \cite{bayati2013estimating}).
	Thus, the proof of the third part follows from (\ref{l3:eq51}) and Condition C4.
	
{\bf Proof of Lemma \ref{lem:tr}:}
The  proof  is similar to the proof of Lemma \ref{Lemma:2}, for the completeness, we give a detail proof here.
By (\ref{v1}),
\begin{equation}
	\begin{aligned}
		&\frac{1}{4n}\l(\|\mathbb{\bar{Y}}-\mathbb{\bar{X}}\hat{\ubeta}^{(j)} -\sqrt{2n}\hat{e}\|_2^2\r) \\
		&\quad+ \lambda_{\beta}^{(j)} \left\| \hat{\ubeta}^{(j)} \right\|_1 + \lambda^{(j)}_{e} \left\| \hat{e} \right\|_1 \\
		&\leq \frac{1}{4n} \left\|\mathbb{\bar{Y}} - \mathbb{\bar{X}} \ubeta^{(j)} - \sqrt{2n} e^* \right\|_2^2 \\
		&\quad+ \lambda_{\beta}^{(j)} \left\| \ubeta^{(j)} \right\|_1 + \lambda^{(j)}_{e} \left\| e^* \right\|_1.
	\end{aligned}
\end{equation}
Denote $\tilde{z}=\hat{\ubeta}^{(j)}-\ubeta^{(j)}, \tilde{v}=\hat{e}^{(j)}-e^{(j)}$, where $e^{(j)}_{1:n}=e^*, e^{(j)}_{(n+1):2n}=0$, let $v^{(1)}=\tilde{v}_{1:n}, v^{(2)}=\tilde{v}_{(n+1):2n}$, it holds that,
\begin{equation}
	\begin{aligned}
		\|\mathbb{\bar{Y}}-\mathbb{\bar{X}}\hat{\ubeta}^{(j)} -\sqrt{2n}\hat{e}\|_2^2 =&
		\|\mathbb{\bar{Y}}-\mathbb{\bar{X}}\ubeta^{(j)} -\sqrt{2n}e^{(j)}\|_2^2 \\
		& - 2 \langle  \mathbb{\bar{Y}}-\mathbb{\bar{X}}\ubeta^{(j)} -\sqrt{n}e, \mathbb{\bar{X}} \tilde{z}+ \sqrt{n}v  \rangle \\
		&+ \|\mathbb{\bar{X}}\tilde{z}+ \sqrt{2n}\tilde{v}\|_2^2\\
		=&\|\mathbb{\bar{Y}}-\mathbb{\bar{X}}\ubeta^{(j)} -\sqrt{2n}e^{(j)}\|_2^2 \\
		& - 2 \langle  \epsilon, \mathbb{X} \tilde{z}+ \sqrt{2n}v^{(1)}   \rangle- 2 \langle  \epsilon^{(\mathrm{S}_j)}, \mathbb{X}^{(\mathrm{S}_j)} z+ \sqrt{2n}v^{(2)}  \rangle \\
		& - \langle  \mathbb{X}\Delta^{(\mathrm{S}_j)}, \mathbb{X} \tilde{z}+ \sqrt{2n}v^{(1)}  \rangle+  \langle  \mathbb{X}^{(\mathrm{S}_j)}\Delta^{(\mathrm{S}_j)}, \mathbb{X}^{(\mathrm{S}_j)} \tilde{z}+ \sqrt{2n}v^{(2)}  \rangle \\
		&+ \|\mathbb{\bar{X}}\tilde{z}+ \sqrt{2n}\tilde{v}\|_2^2.
	\end{aligned}
\end{equation}
Denote $T(j)=\mathrm{S}(\Delta^{(j)}), E=\mathrm{S}(e^*)$. Putting these pieces together, 
\begin{equation}\label{eq:inequality}
	\begin{aligned}
		&\frac{1}{4n} \|\mathbb{\bar{X}}\tilde{z} + \sqrt{2n}\tilde{v}\|_2^2 \\
		&\quad\leq \frac{1}{2n}  \langle  \epsilon, \mathbb{X} \tilde{z}+ \sqrt{n}v^{(1)}  \rangle 
		+ \frac{1}{2n}  \langle  \epsilon^{(\mathrm{S}_j)}, \mathbb{X}^{(\mathrm{S}_j)} \tilde{z}+ \sqrt{n}v^{(2)}   \rangle \\
		&\quad\quad + \frac{1}{4n}\langle  \mathbb{X}\Delta^{(\mathrm{S}_j)}, \mathbb{X} \tilde{z}+ \sqrt{n}v^{(1)}   \rangle 
		- \frac{1}{4n}\langle  \mathbb{X}^{(\mathrm{S}_j)}\Delta^{(\mathrm{S}_j)}, \mathbb{X}^{(\mathrm{S}_j)} \tilde{z}+ \sqrt{n}v^{(2)}   \rangle \\
		&\quad\quad + \lambda_{\beta}^{(j)} (\|\tilde{z}_{T(j)}\|_1 - \|\tilde{z}_{T(j)^c}\|_1) 
		+ \lambda^{(j)}_{e} (\|\tilde{v}_E\|_1 - \|\tilde{v}_{E^c}\|_1) \\
		&\quad\leq \frac{1}{2n} \|\mathbb{X}^\top \epsilon\|_\infty \|\tilde{z}\|_1 
		+ \frac{1}{2n} \|\mathbb{X}^{(\mathrm{S}_j)\top} \epsilon^{(\mathrm{S}_j)}\|_\infty \|\tilde{z}\|_1 \\
		&\quad\quad + \frac{1}{2\sqrt{n}} \|\epsilon\|_\infty \|v^{(1)}\|_1 
		+ \frac{1}{2\sqrt{n}} \|\epsilon^{(\mathrm{S}_j)}\|_\infty \|v^{(2)}\|_1 \\
		&\quad\quad + \lambda_{\beta}^{(j)} (\|\tilde{z}_{T(j)}\|_1 - \|\tilde{z}_{T(j)^c}\|_1) 
		+ \lambda^{(j)}_{e} (\|\tilde{v}_E\|_1 - \|\tilde{v}_{E^c}\|_1) \\
		&\quad\quad + \frac{1}{4n}\|\hat{\Sigma}-\hat{\Sigma}^{(\mathrm{S}_j)}\|_{\infty}\|\Delta^{(\mathrm{S}_j)}\|_1\|\tilde{z}\|_1 \\
		&\quad\quad + \frac{1}{4\sqrt{n}}\left(\|\mathbb{X}\|_{\infty}+\|\mathbb{X}^{(\mathrm{S}_j)}\|_{\infty}\right)\|\Delta^{(\mathrm{S}_j)}\|_1\|\tilde{v}\|_1 \\
		&\quad\leq \left(\frac{1}{2n} \|\mathbb{X}^\top \epsilon\|_\infty + \frac{1}{2n} \|\mathbb{X}^{(\mathrm{S}_j)\top} \epsilon^{(\mathrm{S}_j)}\|_\infty + \lambda_{\beta}^{(j)}\right) \|\tilde{z}_{T(j)}\|_1 \\
		&\quad\quad - \left( \lambda_{\beta}^{(j)} - \frac{1}{2n} \|\mathbb{X}^{(\mathrm{S}_j)\top} \epsilon^{(\mathrm{S}_j)}\|_\infty - \frac{1}{2n} \|\mathbb{X}^\top \epsilon\|_\infty \right) \|\tilde{z}_{T(j)^c}\|_1 \\
		&\quad\quad + \left(\frac{\|\epsilon\|_\infty \vee \|\epsilon^{(\mathrm{S}_j)}\|_\infty }{\sqrt{n}} + \lambda^{(j)}_{e} + \frac{1}{4\sqrt{n}}\left(\|\mathbb{X}\|_{\infty}+\|\mathbb{X}^{(\mathrm{S}_j)}\|_{\infty}\right)\|\Delta^{(\mathrm{S}_j)}\|_1\right) \|\tilde{v}_E\|_1 \\
		&\quad\quad - \left(\lambda^{(j)}_{e} - \frac{\|\epsilon\|_\infty\vee \|\epsilon^{(\mathrm{S}_j)}\|_{\infty}}{\sqrt{n}} - \frac{1}{4\sqrt{n}}\left(\|\mathbb{X}\|_{\infty}+\|\mathbb{X}^{(\mathrm{S}_j)}\|_{\infty}\right)\|\Delta^{(\mathrm{S}_j)}\|_1\right) \|\tilde{v}_{E^c}\|_1.
	\end{aligned}
\end{equation}
By condition C1,
$$
\frac{1}{4\sqrt{n}}\l(\|\mathbb{X}\|_{\infty}+\|\mathbb{X}^{(\mathrm{S}_j)}\|_{\infty}\r)=O_P\l(\sqrt{\frac{\log (np)}{n}} \r),
$$
Thus, by the extended RE condition C1, it holds with probability approach to 1 as $n\to \infty$ that,
\begin{align*}
	\kappa_l^2(\|\tilde{z}\|_2 + \|\tilde{v}\|_2)^2 &\leq  C\sqrt{s_{j,0}}\sqrt{\frac{\log (p)}{n}}\|\tilde{z}\|_2+ C\sqrt{k}\sqrt{\frac{\log (np)}{n}}\|\Delta^{(\mathrm{S}_j)}\|_1 \|\tilde{v}\|_2,
\end{align*}
for some universal constant $C$. Combining these pieces together, we conclude
\begin{equation}\label{lem6:eq1}
\|\tilde{z}\|_2 + \|\tilde{v}\|_2 =O_P\l(\sqrt{s_{j,0}\vee k}\sqrt{\frac{\log (np)}{n}}\|\Delta^{(\mathrm{S}_j)}\|_1 \r).
\end{equation}
By (\ref{eq:inequality}),
\begin{align*}
	&\left(\lambda^{(j)}_{e} - \frac{\|\epsilon\|_\infty\vee \|\epsilon^{(\mathrm{S}_j)}\|_{\infty}}{\sqrt{n}} - \frac{1}{4\sqrt{n}}\left(\|\mathbb{X}\|_{\infty}+\|\mathbb{X}^{(\mathrm{S}_j)}\|_{\infty}\right)\|\Delta^{(\mathrm{S}_j)}\|_1\right) \|\tilde{v}_{E^c}\|_1\\
	&\quad\leq \left(\frac{1}{2n} \|\mathbb{X}^\top \epsilon\|_\infty + \frac{1}{2n} \|\mathbb{X}^{(\mathrm{S}_j)\top} \epsilon^{(\mathrm{S}_j)}\|_\infty + \lambda_{\beta}^{(j)}\right) \|\tilde{z}_{T(j)}\|_1 \\
&\quad\quad + \left(\frac{\|\epsilon\|_\infty \vee \|\epsilon^{(\mathrm{S}_j)}\|_\infty }{\sqrt{n}} + \lambda^{(j)}_{e} + \frac{1}{4\sqrt{n}}\left(\|\mathbb{X}\|_{\infty}+\|\mathbb{X}^{(\mathrm{S}_j)}\|_{\infty}\right)\|\Delta^{(\mathrm{S}_j)}\|_1\right) \|\tilde{v}_E\|_1,
\end{align*}
thus, with probability to 1, we have that
\begin{equation}\label{lem6:eq2}
 \|\tilde{v}\|_1\leq  C\sqrt{s_{j,0}}\sqrt{\frac{\log p}{\log n}} \|\tilde{z}\|_2  + \sqrt{k} \|\tilde{v}\|_2,
\end{equation}
for some universal constant $C$. By the L1 bound of Lasso estimator(Corollary4.5 in \cite{fan2020statistical}), (\ref{lem6:eq1}) and (\ref{lem6:eq2}), 
 \begin{equation}\label{lem6:eq3}
 	\begin{aligned}
|\hat{h}_j - \|\Delta^{(j)}\|_1| &\leq 2\left\|\hat{\ubeta}^{(j)} -\ubeta^{(j)}\right\|_1+2\left\| \hat{\ubeta}^{(\mathrm{S}_j)}-\ubeta^{(\mathrm{S}_j)}\right\|_1\\
&\leq 2\left\|\hat{\ubeta}^{(j)} -\ubeta^{(j)}\right\|_1+2\left\| \hat{\ubeta}^{(\mathrm{S}_j)}-\ubeta^{(\mathrm{S}_j)}\right\|_1\\
&=O_P\l(s_{j,0}\vee k\sqrt{\frac{\log p}{\log n}}\sqrt{\frac{\log (np)}{n}}\|\Delta^{(\mathrm{S}_j)}\|_1 +s_{j,0}\sqrt{\frac{\log p}{n}}\r),
\end{aligned}
 \end{equation}
which completes the proof.
	
	{\bf Proof of Theorem \ref{th:1}:}
	Theorem (\ref{th:1} is a direct consequence of proposition \ref{prop0} and Lemma \ref{lem:tr}.
	\bibliography{ref}
\end{document}